\documentclass[11pt, reqno]{amsart}
\usepackage{todonotes}
\usepackage{xifthen}
\usepackage[T1]{fontenc}
\usepackage[english]{babel}
\usepackage{amsmath, amssymb, amsthm}  

\numberwithin{equation}{section}

\usepackage{mathrsfs}          
\usepackage{amstext, amsfonts, a4}
\usepackage{bbm}
\usepackage{wasysym}
\usepackage{caption}
\usepackage{lmodern}
\usepackage{booktabs} % For formal tables
\usepackage{mathtools}
\mathtoolsset{showonlyrefs}

\usepackage{multirow}

\newcounter{dummy}
\usepackage{enumitem}
\makeatletter
\newcommand\myitem[1][]{\item[#1]\refstepcounter{dummy}\def\@currentlabel{#1}}
\makeatother

\usepackage{animate} % for gif

\usepackage[
backend=biber,        % compilateur par dÃ©faut pour biblatex
sorting=nty,          % trier par nom, annÃ©e, titre
maxcitenames=2,
mincitenames=1,
maxbibnames=1000,
isbn=false,
doi=false,
eprint=true,
url=false,
giveninits=true
]{biblatex}

\DeclareMathSymbol{\shortminus}{\mathbin}{AMSa}{"39}

\usepackage[colorlinks=true]{hyperref}
\hypersetup{urlcolor=blue, citecolor=blue, linkcolor=blue}

\usepackage{tikz}
\usetikzlibrary{positioning}

\usepackage{color,transparent}
\usepackage{graphics,graphicx,subfigure}
\usepackage{csquotes}
\usepackage{comment}
\usepackage{float} 
\usepackage{stmaryrd}

\theoremstyle{plain}
\newtheorem{theorem}{Theorem}[section] %compteur commenÃ§ant par le numéro de la section (on pourrait aussi faire commencer par le numéro de la sous-section - remplacer "section" par "subsection")
\newtheorem{proposition}[theorem]{Proposition}%[section]      
\newtheorem{lemma}[theorem]{Lemma}%[section]            %etc...
\newtheorem{corollary}[theorem]{Corollary}%[section]

\theoremstyle{definition}
\newtheorem{definition}{Definition}
\newtheorem{remark}[theorem]{Remark}
\def\ogg~{{\rm \og}}   % guillemets ouvrants
\def\emptyset{\varnothing}

\def\N{{\mathbb N}}    %naturels
\def\Z{{\mathbb Z}}     %entiers relatifs
\def\R{{\mathbb R}}    %rÃ©els
\def\C{{\mathbb C}}    %complexes
    \def\cS{{\mathcal S}}                     \def\cF{{\mathcal F}}    \def\cX{{\mathcal X}}   \def\cZ{{\mathcal Z}}

\newcommand\Res{\operatorname{Res}}

\hypersetup{
	pdfauthor={Cyprien Tamekue},
	pdfsubject={Paper manuscrit},
	pdftitle={On Billock and Tsou's psychophysical experiments modeling},
	pdfkeywords={neural field model, Billock and Tsou's experiments, visual illusions and perceptions, Spatially forced pattern forming system}
}

\begin{document}
	\title{Reproducibility via neural fields of visual illusions induced by localized stimuli}
	% \title[Reproducibility of Billock and Tsou's Experiments]{Reproducibility of Billock and Tsou's Experiments via neural fields: Analyzing the interplay of excitatory and inhibitory responses}
	\thanks{This work has been supported by the ANR-20-CE48-0003 and by the ANR-11-IDEX-0003, Objet interdisciplinaire H-Code. The first author was supported by a grant from the bourse de thèses ``Jean-Pierre Aguilar''. Corresponding author: Cyprien Tamekue.}
	
	\author{Cyprien Tamekue}
	\address{Université Paris-Saclay, CNRS, CentraleSupélec, Laboratoire des signaux et systèmes, 91190, Gif-sur-Yvette, France}
	\email{cyprien.tamekue@centralesupelec.fr}
	
	\author{Dario Prandi}
	\address{Université Paris-Saclay, CNRS, CentraleSupélec, Laboratoire des signaux et systèmes, 91190, Gif-sur-Yvette, France}
	\email{dario.prandi@centralesupelec.fr}
	
	\author{Yacine Chitour}
	\address{Université Paris-Saclay, CNRS, CentraleSupélec, Laboratoire des signaux et systèmes, 91190, Gif-sur-Yvette, France}
	\email{yacine.chitour@centralesupelec.fr}
	
	\begin{abstract}
		This paper focuses on the modeling of experiments conducted by Billock and Tsou [V. A. Billock and
		B. H. Tsou, Proc. Natl. Acad. Sci. USA, 104 (2007), pp. 8490--8495] using an Amari-type neural field that models the average membrane potential of neuronal activity in the primary visual cortex (V1). The study specifically focuses on a regular funnel pattern localized in the fovea or the peripheral visual field. It aims to comprehend and model the visual phenomena induced by this pattern, emphasizing their nonlinear nature. The research involves designing sensory inputs that mimic the visual stimuli from Billock and Tsou's experiments. The cortical outputs induced by these sensory inputs are then theoretically and numerically studied to assess their ability to model the experimentally observed visual effects at the V1 level. A crucial aspect of this study is the exploration of the effects induced by the nonlinear nature of neural responses. By highlighting the significance of excitatory and inhibitory neurons in the emergence of these visual phenomena, the research suggests that an interplay of both types of neuronal activities plays a crucial role in visual processes, challenging the assumption that the latter is primarily driven by excitatory activities alone.
		
		\textbf{MSCcodes.}  92C20, 35B36, 45A05, 45G15, 45K05, 	65R20.
	\end{abstract}

	\maketitle
	\tableofcontents
	
	\section{Introduction}
	
	Exploring a mathematically sound approach to understanding visual illusions in human perception using neural dynamics can give us valuable insights into perceptual processes and visual organization \cite{bertalmio2020visual,bertalmio2021cortical}, and can reveal much about how precisely the brain works. Neural dynamics refers to the patterns of activity and interactions among neurons that give rise to our ability to see and understand the world. Our visual system processes information in different stages, with specialized neurons at each stage extracting specific details from what we see. The visual system shows dynamic and widespread activity patterns, from detecting basic features like edges and orientations to putting everything together and making sense of it.
	The brain area which detects basic features such as spatial position, edges, local orientations, and direction in visual stimuli from the retina is the primary visual cortex (V1 for short), \cite{hubel1959,hubel1977ferrier}. 
	
	Simple geometric visual hallucinations akin to that classified by Klüver \cite{kluver1966} have been theoretically recovered in the last decades via the neural dynamic equation used to model the cortical activity in V1 combined with the bijective nonlinear retino-cortical mapping \cite{schwartz1977,tootell1982} between the visual field and V1, see for instance, \cite{bressloff2001,bressloff2002,ermentrout1979,golubitsky2003,tass1995}. These geometric forms, known as form constants, are obtained near a Turing-like instability using linear stability analysis, (equivariant) bifurcation theory, and pattern selection when the cortical activity is due solely to the random firing of V1 neurons, that is, in the absence of sensory inputs from the retina. However, to function correctly, the primary visual cortex must be primarily driven by sensory information from the retina \cite{hubel1959,hubel1977ferrier}, not only by the internal noisy fluctuation of its cells. { Several methods have explored how sensory inputs are processed in early visual areas. Experimental studies have been conducted \cite{hebb2005organization}, along with experimentally induced phenomena via psychophysical tests \cite{billock2007,billock2012elementary,rogers2021hallucinations,pearson2016sensory,mackay1957,mackay1961}. Additionally, theoretical tools like the Lie transformation group model have been applied to analyze perceptual processes \cite{dodwell1983lie,hoffman1966lie}. Despite these efforts, using theoretical neural dynamics, our understanding of the precise neuronal mechanisms underlying visual illusions remains elusive.}
	\begin{figure}[t]
		\centering
		\includegraphics[width=.6\linewidth]{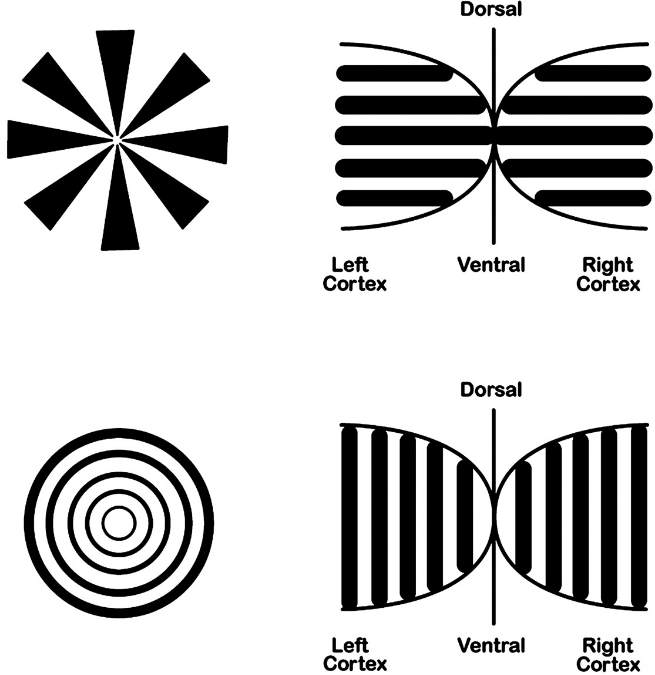}
		\caption{Visual illustration of the retino-cortical map, redrawn from \cite{billock2007}. The \emph{top-left} corresponds to the funnel pattern in the retina, and on the \emph{top-right}, the corresponding pattern of horizontal stripes is in V1. While the \emph{bottom-left} corresponds to the tunnel pattern in the retina, and on the \emph{bottom-right}, the corresponding pattern of vertical stripes is in V1. In particular, these images are regular in shape and symmetrical with respect to a specific subgroup of the plane's motion group \cite{bressloff2001}. 
		}
		\label{fig:retino-cortical}
	\end{figure}
	It has been known since Helmholtz's work \cite{helmholtz1867} that even simple geometrical patterns comprising black and white zones may induce strong after-images\footnote{{ In the experiments of \cite{billock2007}, observers perceive an illusory image in their visual field after viewing a visual stimulus, and this image persists for a few seconds. This is what we refer to when using the term \textit{after-image}.
	}} 
	accompanying a visual perception after a few seconds. Then, via redundant and non-redundant stimulation by funnel (fan shape) and tunnel (concentric rings) patterns (see Figure~\ref{fig:retino-cortical}), MacKay \cite{mackay1957,mackay1961} points out that there is some kind of orthogonal response in the visual cortex since funnel pattern induces a superimposed (to the physical stimulus) tunnel pattern as an after-image, and conversely.
	
	{
		More recently, Nicks \emph{et al.} \cite{nicks2021} have built on the foundation of neural field equations of Amari-type \cite[Eq. (3)]{amari1977} to model MacKay-type visual illusions induced by specific visual stimuli at the cortical level. Their model, which represents cortical activity in V1, incorporates a fully distributed state-dependent sensory input. 
		This input models the cortical representation via the retino-cortical map of funnel and tunnel patterns. 
		Theoretically, they proved these experimental findings, demonstrating an orthogonal response of V1 to visual inputs. The present authors have further sustained this evidence in their previous works \cite{tamekue2024mathematical,tamekue2022reproducing}.
	}
	In particular, by using the neural field equation of Amari-type, we have shown that the underlying Euclidean symmetry of V1 (see, for instance, \cite{bressloff2001}) restricts the geometrical shape of visual inputs that can induce a ``strong'' after-effect in the primary visual cortex. If the visual input is symmetric with respect to a subgroup of the group of the motion of the plane {(refer to \cite[Appendix~A]{tamekue2024mathematical})}, then the induced after-image obtained via the Amari-type equation and the inverse retino-cortical map have the same subgroup as a group of symmetry. The latter suggests that the after-images induced by fully distributed tunnel and funnel patterns (more generally spontaneous patterns obtained through Turing-like instability \cite{bressloff2001,ermentrout1979,tass1995}) that fill all the visual field have the same shape. Moreover, we exhibited in \cite{tamekue2023,tamekue2022reproducing} numerical simulations using the Amari-type equation, showing that if the funnel pattern is localized either in the fovea (centre of the visual field) or in the peripheral visual field, then the induced after-image consisting of the tunnel pattern appears in the white or black complementary region where the stimulus is not localized--also demonstrating orthogonal and non-local response--in V1. These numerical simulations, therefore, sustain the psychophysical experiments reported by Billock and Tsou \cite{billock2007}, see also \cite{billock2012elementary}. Note that numerical simulations (including those for rotating after-images that are not considered in this paper) performed in \cite{nicks2021} also support the latter psychophysical experiments.
	\begin{figure}[t]
		\centering
		% \vspace{-1cm}
		\includegraphics[width=.3\linewidth]{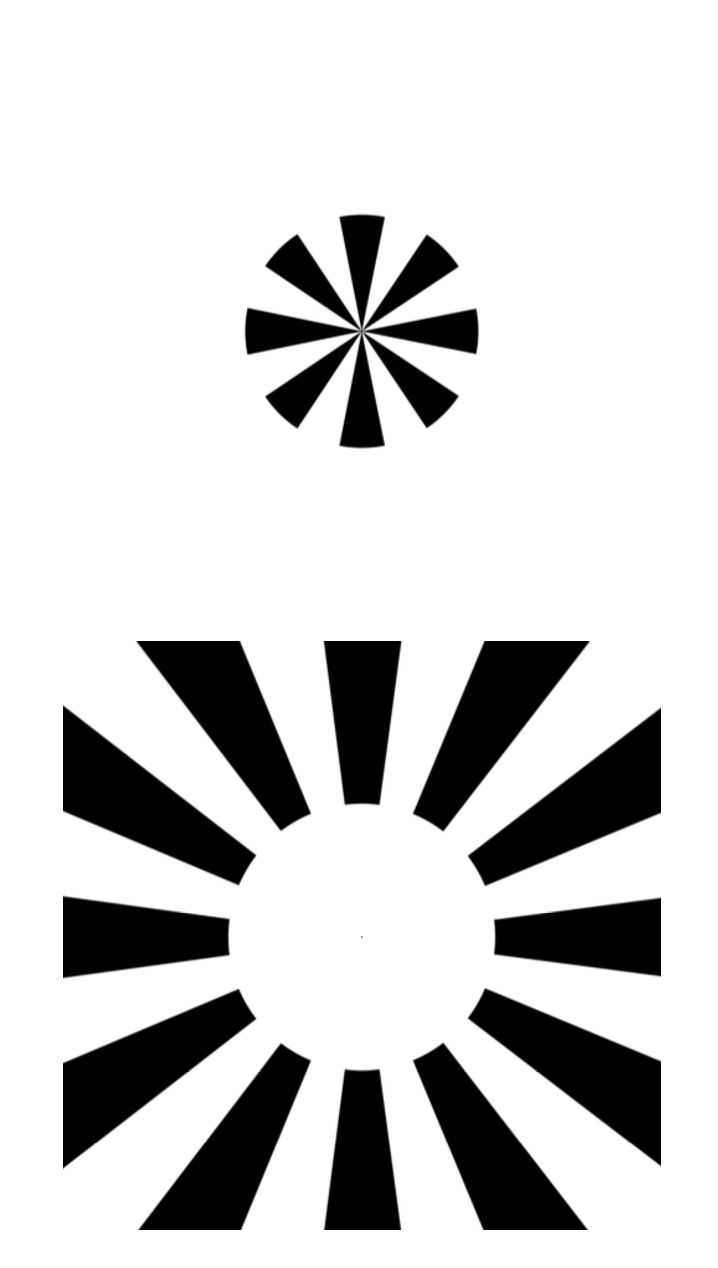}\hspace{4.5em}
		\includegraphics[width=.3\linewidth]{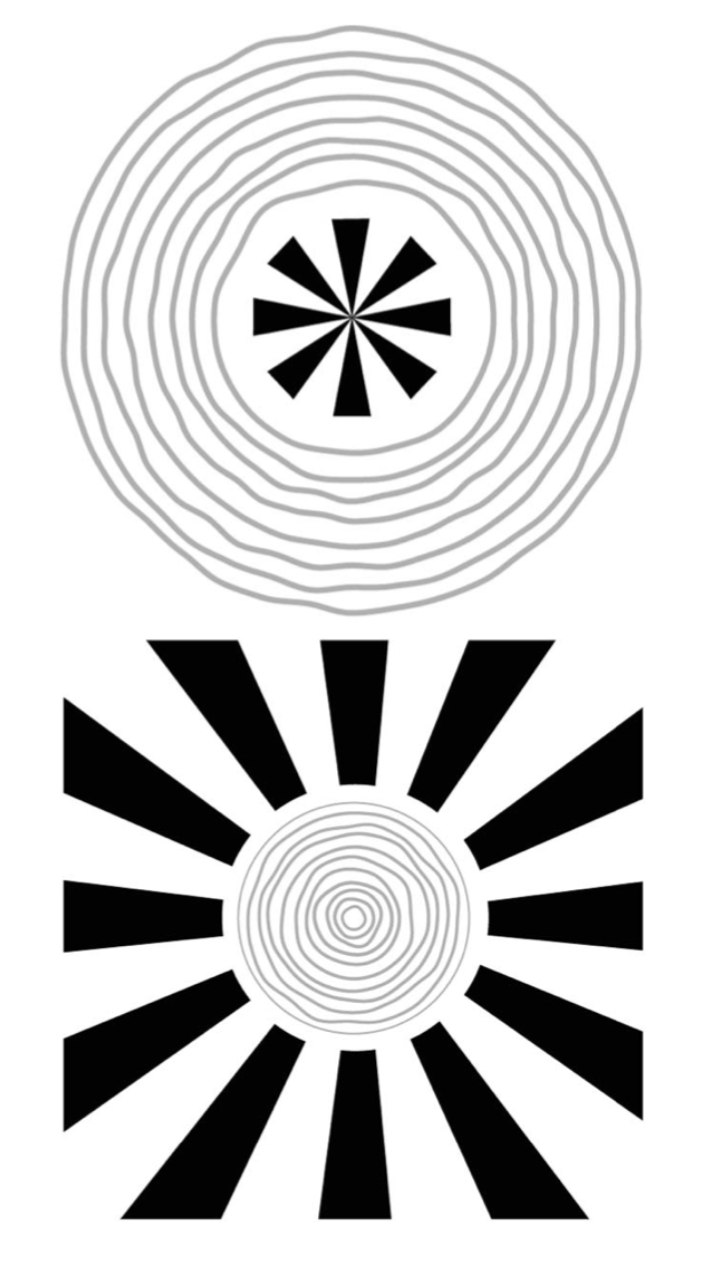}
		\caption{Billock and Tsou's experiments: the presentation of a funnel pattern stimulus in the centre (image on the \textit{top-left}) induces an illusory perception of tunnel pattern in surround (image on the \emph{top-right}) after a flickering of the empty region (the white region surrounding the stimulus pattern on the \textit{top-left}). We have a reverse effect on the \emph{bottom}. Adapted from \cite[Fig.~3]{billock2007}.}%Adapted from \parencite{mackay1957,zeki1993}.
		\label{fig:Billock-Tsou}%bottom and top
	\end{figure}
	\subsection{Billock and Tsou's psychophysical experiments}\label{ss::BT experiments}
	
	Significant visual effects associated with funnel and tunnel patterns have been recently observed in the psychophysical experiments conducted by Billock and Tsou \cite{billock2007}. Like the MacKay effect \cite{mackay1957,mackay1961}, these authors discovered that introducing biased stimuli elicits orthogonal responses in the visual field. When a physical stimulus is localized at the fovea (the central region of the visual field), the resulting visual illusion appears in the flickering periphery. Conversely, the visual illusion emerges in the flickering centre if the physical stimulus is presented in the periphery.
	Specifically, when a background flicker is combined with a funnel pattern centered on the fovea (or periphery), the observer experiences the illusory perception of a tunnel pattern in the periphery (or fovea, respectively). Similarly, when the periphery (or fovea) of a tunnel pattern localized at the fovea (or periphery) is subjected to flickering, an illusory rotating funnel pattern is perceived in the periphery (or fovea).
	In all cases, the illusory contours in the afterimage appear within the nonflickering region, depending on whether the flicker does not extend through the physical stimulus or if the empty region is flickered out of phase. Refer to Fig.~\ref{fig:Billock-Tsou} for a visual illustration.

	\subsection{Neural field model, strategy of study, and presentation of our results }
	This paper aims to investigate the theoretical modeling of Billock and Tsou's experiments~\cite{billock2007} associated with a regular funnel pattern localized in the fovea or peripheral visual field, as recalled in the previous section. 
	{
		Our approach is mechanistic: we describe a possible model of how cortical dynamics induce the phenomena under consideration. The matter of \textit{why} neurons behave this way is outside the scope of this article, albeit being a very active topic of investigation in theoretical neuroscience \cite{CLIFFORD20073125,laparra2015,rentzeperis2023}
		.}
	We will follow the idea of controllability of the Amari-type neural field introduced in \cite{tamekue2024mathematical,tamekue2022reproducing} {that we will recall hereafter}. In particular, we stress why these visual phenomena are nonlinear, as first pointed out in \cite{tamekue2023}. 
	
	% {We recall the idea of controllability right here.}
	From a control theory point of view, the first aim is to design a suitable sensory input $I$, V1 representation via the retino-cortical map of visual stimulus from the retina used in the experiment such that the cortical state $a:\R_{+}\times\mathbb R^2\to \R$ solution to the following Amari-type control system
	\begin{equation}\label{eq::NF-intro}\tag{NF}
		\begin{split}
			\partial_t a(x,t) + a(x,t) - \mu\int_{\R^2}\omega(x-y)f(a(y,t))dy &= I(x),\qquad (t,x)\in\R_{+}\times\mathbb R^2\\
			a(x,0) &= a_0(x),\qquad x\in\R^2
		\end{split}
	\end{equation}
	exponentially stabilizes to the stationary state, corresponding to the V1 representation via the retino-cortical map of the induced after-image reported by Billock and Tsou. 
	{ Following \cite{bressloff2001,ermentrout1979}, we assume that the perceived image is obtained by applying the inverse retino-cortical map to the cortical state.}
	Secondly, we will perform a quantitative and qualitative study of this stationary output to show that it captures all the essential features of the visual illusion announced by Billock and Tsou at the V1 level. 
	To this aim, we follow a numerical analysis approach specifically designed to address the complex nonlinear dynamics characteristic of the considered neural fields model.
	
	{
		Equation~\eqref{eq::NF-intro} has been introduced in \cite{amari1977} (see also \cite{cook2022neural} for a recent overview on neural field models) to describe the dynamics of the average membrane potential of the neurons located at the point $x\in\R^2$ at time $t\ge 0$.
		% , which is the electrical potential difference across the cell membrane. 
		The time-evolution of the average membrane potential at time $t\ge 0$ is given by the map $x\mapsto a(x,t)$.
		The state of cortical activity in V1 at time $t=0$ is assumed to be given by the function $a_0$.
	}
	{
		The neural field equation \eqref{eq::NF-intro} can be seen as the combined action of the external stimulus $I$, the  natural decay rate or leakage of the neurons' average membrane potential towards their resting states, 
		% natural exponential decay of the neural activity, 
		and an integral term representing the intra-neural connectivity, modulated by the parameter $\mu>0$.
		The latter consists of a convolution product between the synaptic connectivity kernel $\omega$, modeling the spatial relationship between neurons, and the non-linear term given by the response function $f$ applied to $a$, which transforms the activity level of a neural population at location $y$ and time $t$ into an output signal. (See Definition~\ref{def:response-fct}.) }
	Therefore, once the signal reaches V1, it will interact with local neural dynamics captured by this equation. The equation then models how V1 responds to this input while accounting for local interactions (via the connectivity kernel $\omega$) and nonlinearities in neural activity (via the response function  $f$).
	
	In biological brain tissue, neurons can be excitatory or inhibitory \cite{hubel1959,hubel1977ferrier}, and an inhibitory neuron decreases the likelihood that a post-synaptic neuron will send out electrical signals or spike to communicate with other brain cells. A negative value for $f(a(y,t))$ might capture this inhibitory influence. Notice also that a nonnegative function $f(a(y,t))\ge 0$ would imply that all neurons, regardless of their current activity level, provide some excitatory output. This overlooks the crucial role of inhibitory neurons in shaping neural activity. Moreover, as evident from the study we will present in this paper, a model lacking inhibitory activity is likely insufficient for capturing certain phenomena such as that reported by Billock and Tsou. In the latter case, we will also see that a complex interplay between excitatory and inhibitory activity \cite{haider2006neocortical,shu2003turning} in the shape of $f$ is required and plays a crucial role since an odd nonlinearity does not model the phenomenon.
	
	Therefore, the effect that plays the non-linearity $f$ on the reproducibility of Billock and Tsou's experiments using \eqref{eq::NF-intro} will be highlighted. As we previously pointed out in \cite[Fig.~8]{tamekue2023}, these phenomena are wholly nonlinear and strongly depend on the shape of the nonlinear function $f$. 
	
	{ Different models for cortical activity are available (e.g., the original Wilson-Cowan model for excitatory/inhibitory populations \cite{wilson1973}). The choice of an Amari-type neural field is motivated by the fact that equation \eqref{eq::NF-intro} is sufficient for describing the spontaneous formation of funnel and tunnel patterns \cite[Eq.~(16)]{bressloff2001} in V1, and we expect it also to be suitable for describing psychophysical experiments involving these patterns. Moreover, it is more amenable to mathematical analysis and yields the same qualitative behaviors as the original Wilson-Cowan model in many neural processes modeling \cite{ermentrout1998neural}.}
	{
		We also mention that the Amari-type neural field has been successfully applied to reproduce visual illusions \cite{baspinarCorticalInspired2021,bertalmio2021cortical}, and has recently been connected with the Divisive Normalization from visual psychophysics \cite{maloCortical2024}.
	}

	{ In this work, we focus on fixed contrast stimuli. We stress that as a consequence of our results, one can easily obtain that, for a fixed nonlinearity, the illusory phenomena are not reproduced for small contrast.}
	
	Notice that while sensory inputs in Billock and Tsou's experiments are time-varying, our study finds that a temporal flicker of the complementary region where the stimulus is not localized is not necessary to reproduce these intriguing visual phenomena (an observation already made in \cite{nicks2021}).
	Our interpretation is that Billock and Tsou's phenomena result wholly from the underlying non-local and nonlinear properties of neural activity in V1 rather than the temporal flickering of the complementary region where the stimulus is not localized. In particular, the flickering should instead be in the origin of illusory motions that subjects perceived in the after-images.% in these experiments
	
	The remaining of the paper is organized as follows: Section~\ref{s::GN} recalls some general notations used throughout the following. We present assumptions on model parameters used in \eqref{eq::NF-intro} in Section~\ref{ss::Assumption on parameters}. Section~\ref{ss::mathematical modeling} describes the mathematical modeling of visual stimuli associated with funnel patterns used in Billock and Tsou's experiments. In Section~\ref{s::well-posedness}, we recall some preliminary results related to the well-posedness of equation \eqref{eq::NF-intro} and those in the direction of modeling Billock and Tsou's experiments associated with a funnel pattern localized either in the fovea or in the peripheral visual field. The modeling of the phenomena using \eqref{eq::NF-intro} starts precisely in Section~\ref{s::theoretical Billock and Tsou}. In Section~\ref{ss::theoretical Billock and Tsou experiments linear}, we prove that the stationary output of \eqref{eq::NF-intro} associated with a pattern of horizontal stripes localized in the left area of V1 does not contain a pattern of vertical stripes in the white complementary region (the right area of V1) but rather a mixture of horizontal and vertical stripes if the response function is linear. In Section~\ref{ss::unreproducibility of BT with certain nonlinearities}, we prove that even with certain nonlinear response functions that exhibit strong inhibitory or excitatory influences and a weak slope, or a balance between excitatory and inhibitory influences, the stationary output of \eqref{eq::NF-intro} associated with a pattern of horizontal stripes localized in the left area of V1 is identical with that of the linear response function. Section~\ref{s::numerics} focuses precisely on proving that if, for instance, the response function in \eqref{eq::NF-intro} exhibits a good interplay between excitatory and inhibitory influence and a weak slope, then the stationary output associated with a pattern of horizontal stripes localized in the left area of V1 contains a pattern of vertical stripes in the white complementary region (the right area of V1) as Billock and Tsou reported. For this aim, we follow a numerical analysis-type of argument in Section~\ref{s::Analysis of NS}, together with an analysis of the corresponding numerical schemes. Section~\ref{s::Simulations} presents some numerical simulations that bolster our theoretical study. Finally, in Section~\ref{s::concluding remarks and discussion}, we discuss the main results of our paper and highlight areas for future work. We defer to  Appendix~\ref{s::complement resluts}, the proof of some technical results used in the paper.

	\subsection{General notations}\label{s::GN}
	
	In the following, $d\in\{1,2\}$ is the dimension of $\R^d$ and $|x|$ denote the Euclidean norm of $x\in\R^d$. For $p\in\{1,\infty\}$, $L^p(\R^d)$ is the Lebesgue space of class of real-valued measurable functions $u$ on $\R^d$ such that $|u|$ is integrable over $\R^d$ if $p=1$, and $|u|$ is essentially bounded over $\R^d$ when $p=\infty$. We endow these spaces with their standard norms $\|u\|_1 = \int_{\R^d}|u(x)|dx$  and $\|u\|_\infty = \operatorname{ess}\sup_{x\in\R^d}|u(x)|$. 
	We let %$X_\infty:=
	$C([0,\infty);L^\infty(\R^d))$
	be the space of all real-valued functions $u$ on $\R^d\times[0,\infty)$ such that, $u(x,\cdot)$ is continuous on $[0,\infty)$ for $\mbox{a.e.},\; x\in\R^d$ and $u(\cdot,t)\in L^p(\R^d)$ for every $t\in[0,\infty)$. We endow this space with the norm $\|u\|_{L_t^\infty L_x^\infty} = \sup\limits_{t\ge 0}\|u(\cdot,t)\|_\infty$.
	
	We let $\cS(\R^d)$ be the Schwartz space of rapidly-decreasing $C^\infty(\R^d)$ functions,
	and $\cS'(\R^d)$ be its dual space, i.e., the space of tempered distributions. Then, $\cS(\R^d)\subset L^p(\R^d)$ and $L^p(\R^d)\subset\cS'(\R^d)$ continuously. 
	The Fourier transform of $u\in L^1(\R^2)$ is defined by
	\begin{equation}\label{eq::Fourier transform in S}
		\widehat{u}(\xi):= \cF\{u\}(\xi)=\int_{\R^d} u(x)e^{-2\pi i\langle x,\xi\rangle}dx,\quad\forall\xi\in\R^d.
	\end{equation}
	Since $\cS(\R^d)\subset L^1(\R^2)$, one can extend the above by duality to $\cS'(\R^d)$, and in particular to $L^\infty(\mathbb{R}^d)$. The convolution of $u\in L^1(\R^d)$ and $v\in L^p(\R^d)$, $p\in\{1,\infty\}$, is %defined by
	\begin{equation}\label{eq::spatial convolution}
		(u\ast v)(x) = \int_{\R^d}u(x-y)v(y)dy,\qquad\forall x\in\R^d.
	\end{equation}
	
	Finally, the following notation will be helpful: if $F$ is a real-valued function defined on $\R^2$, we use  $F^{-1}(\{0\})$ to denote the zero level-set of $F$.
	
	\section{Assumption on parameters and mathematical modeling of visual stimuli }
	In this section, we will present assumptions that we will consider on the parameters in \eqref{eq::NF-intro}, specifically on the response function $f$ and on the connectivity kernel $\omega$, as it is highlighted in Section~\ref{ss::Assumption on parameters}. Then, in Section~\ref{ss::mathematical modeling}, we will present how we mathematically model the visual stimuli used in Billock and Tsou's experiments associated with a regular funnel pattern localized in the fovea or peripheral visual field that we incorporate as sensory inputs in \eqref{eq::NF-intro}. 
	
	\subsection{Assumption on parameters in the Amari-type equation}\label{ss::Assumption on parameters}
	We make the following assumption on parameters involved in \eqref{eq::NF-intro}. 
	\newline
	\paragraph{\textbf{Coupling kernel}:} In this article, we use a spatially homogeneous and isotropic interaction kernel $\omega$ in relation to coordinates $(x_1,x_2)$. It depends solely on the Euclidean distance among neurons, showing rotational symmetry. The ``Mexican-hat'' distribution is employed, a variant of the Difference of Gaussians (DoG) model with two components. The first Gaussian governs short-range excitatory interactions, and the second Gaussian models long-range inhibitory interactions in V1 neurons. Thus, the connectivity kernel is taken as:
	\begin{equation}\label{eq::connectivity}
		\omega(x) = [2\pi\sigma_1^2]^{-1}e^{-\frac{|x|^2}{2\sigma_1^2}}-\kappa[2\pi\sigma_2^2]^{-1}e^{-\frac{|x|^2}{2\sigma_2^2}},\qquad x\in\R^2
	\end{equation}
	where $\kappa>0$, and $\sigma_1$ and $\sigma_2$ satisfy $0<\sigma_1<\sigma_2$ and $\sigma_1\sqrt{\kappa}<\sigma_2$. The latter condition is crucial for explicitly calculating the $L^1$-norm of $\omega$, as detailed in \eqref{eq::L^1-norm of omega}.
	
	Note that $\omega(x)$ is equivalent to $\omega(|x|)$, and $\omega$ belongs to the Schwartz space $\cS(\R^2)$. The Fourier transform of $\omega$ is explicitly given by:
	\begin{equation}\label{eq::Fourier transform of the kernel omega}
		\widehat{\omega}(\xi)= e^{-2\pi^2\sigma_1^2|\xi|^2}-\kappa e^{-2\pi^2\sigma_2^2|\xi|^2},\qquad\qquad\qquad\forall\xi\in\R^2
	\end{equation}
	and the maximum of $\widehat{\omega}$ occurs at every vector $\xi_c\in\R^2$ satisfying $|\xi_c| = q_c$. Thus:
	\begin{equation}\label{eq::maximum of omegahat}
		q_c := \sqrt{\frac{\log\left(\frac{\kappa\sigma_2^2}{\sigma_1^2}\right)}{2\pi^2(\sigma_2^2-\sigma_1^2)}}\qquad\qquad\text{and}\qquad\qquad \max\limits_{r\ge 0}\widehat{\omega}(r) = \widehat{\omega}(q_c).
	\end{equation}
	The $L^1$-norm of $\omega$ is also explicitly represented by:
	\begin{equation}\label{eq::L^1-norm of omega}
		\|\omega\|_1 = (1-\kappa)+2\left(\kappa e^{-\frac{\Theta^2}{2\sigma_2^2}}-e^{-\frac{\Theta^2}{2\sigma_1^2}}\right)\qquad\quad\text{with}\qquad\quad \Theta:= \sigma_1\sigma_2\sqrt{\frac{2\log\left(\frac{\sigma_2^2}{\kappa\sigma_1^2}\right)}{\sigma_2^2-\sigma_1^2}}.
	\end{equation}
	
	Let us mention that $\omega$ might not satisfy the \textit{balanced} condition $\widehat{\omega}(0)=0$, an equilibrium between excitation and inhibition. Nonetheless, this equilibrium is achieved when $\kappa = 1$.
	\begin{figure}[t]
		\centering
		\includegraphics[width=.45\linewidth]{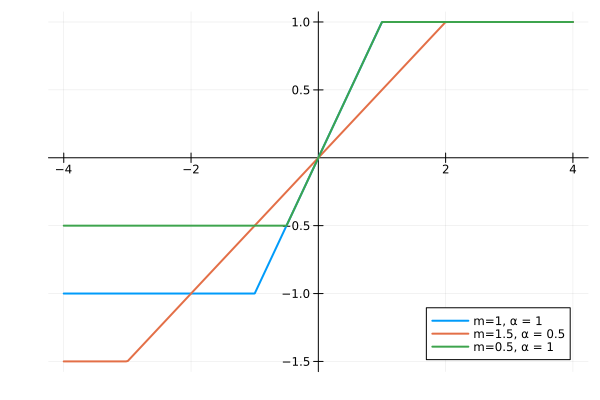}\hspace{0.1em}
		\includegraphics[width=.45\linewidth]{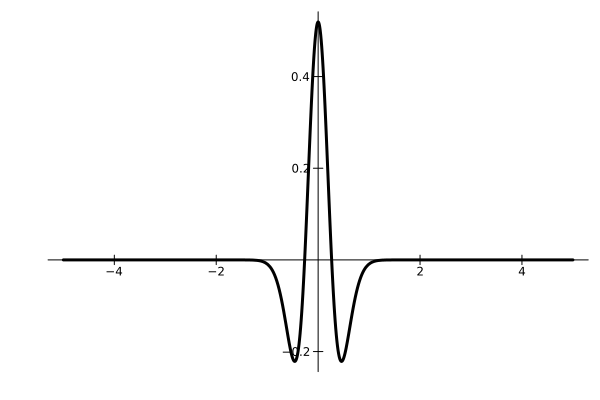}
		\caption{On the \textit{left}, nonlinear response functions $f_{m,\alpha}(s) = \max(-m,\min(1,\alpha s))$ for different values of $m$ and $\alpha$. On the \textit{right} a $1$D DoG kernel $\omega$.}
		\label{fig:response function}
	\end{figure}
	
	Finally, in the sequel, we use the letter $C_\omega$ to denote any positive constant depending only on the parameters involved in the definition of $\omega$.
	
	\paragraph{\textbf{Response function}:} 
	The choice of the response function $f$ is crucial, and it is motivated by authors' previous works \cite{tamekue2023,tamekue2022reproducing}. Indeed, in \cite[Figs.~5 and~6]{tamekue2022reproducing} we illustrated the capability of Equation \eqref{eq::NF-intro} to reproduce Billock and Tsou experiments with the nonlinear response function $f(s) = (1+\exp(-s+0.25))^{-1}-(1+\exp(0.25))^{-1}$, and that $f(s)=\tanh(s)$ does not reproduce the phenomenon, suggesting that certain (non-odd) sigmoidal-type response functions are required to model the phenomenon. In \cite[Section~4]{tamekue2023}, we briefly explained why the stationary output pattern of the Amari-type \eqref{eq::NF-intro} does not capture the essential features of the visual illusions reported by Billock and Tsou's when the response function is linear. Moreover, still in \cite[Fig.~8]{tamekue2023}, by considering the ``sigmoidal-type'' response function $f_{m,\alpha}(s) = \max(-m,\min(1,\alpha s))$  with $m\ge 0$ and $\alpha>0$, we figured out ranges on parameters $m$ and $\alpha$ for which the stationary output pattern of the Amari-type \eqref{eq::NF-intro} captures the essential features of the visual illusions reported by Billock and Tsou's. More precisely, \cite[Fig.~8]{tamekue2023} suggests that nonnegative $f_{0,\alpha}$, odd $f_{1,\alpha}$ with $0<\alpha<\infty$, nonlinearities $f_{m,\alpha}$ with strong inhibitory influence $m>1$ and weak slope $0<\alpha< 1$ as well as nonlinearities $f_{m,\alpha}$ with strong excitatory influence and weak slope $0<\alpha< m\le 1$ do not model Billock and Tsou's experiments associated with a regular funnel pattern localized either in the fovea or in the peripheral visual field. While for other values of $m$ and $\alpha$, \eqref{eq::NF-intro} with the response function $f_{m,\alpha}$ captures the essential features of the visual illusions reported by Billock and Tsou (either the ``strong'' or the ``weak'' phenomenon, as recalled in Section~\ref{ss::BT experiments}).
	
	Observe also that $f_{m,\alpha}$ is a non-smooth ``mathematical approximation'' of the following sigmoid function, frequently used in neural field models like \eqref{eq::NF-intro},
	\begin{equation}
		\label{eq:sigmoide-BT-smooth}
		g_{\gamma, \nu}(s) := (1+\exp(-\gamma(s-\nu)))^{-1}-(1+\exp(\kappa\nu))^{-1},\qquad \gamma>0, \nu>0.
	\end{equation}
	
	In this paper, when referring to a response function, we will always assume the following.
	
	\begin{definition}
		\label{def:response-fct}
		A response function is a non-decreasing Lipschitz continuous function $f:\R\to \R$ such that $f(0)=0$, $f$ is differentiable at $0$, and $\alpha:=f'(0) = \|f'\|_\infty$.
	\end{definition}
	
	Of particular interest in the rest of the paper is the family of response functions given by
	\begin{equation}\label{eq::f m alpha}
		% f(s):=
		f_{m,\alpha}(s) = \max(-m,\min(1,\alpha s)) = \begin{cases}
			1&\quad\mbox{if}\quad s\ge \frac 1\alpha\cr
			\alpha s&\quad\mbox{if}\quad-\frac m\alpha\le s\le \frac 1\alpha,\cr
			-m&\quad\mbox{if}\quad s\le -\frac m\alpha
		\end{cases}\qquad s\in\R
	\end{equation}
	for every $0\le m<\infty$ and $0<\alpha<\infty$, or by
	\begin{equation}\label{eq::f infinity alpha}
		f_{\infty,\alpha}(s) 
		=\min(1,\alpha s)
		\qquad s\in\R
	\end{equation}
	for every $0<\alpha<\infty$. Please refer to Figure~\ref{fig:response function} for a visual illustration. 
	Notice that, whenever $m\ge 0$ is finite, $f_{m,\alpha}$ is bounded.
	
	Finally, it is worth emphasizing that the spatially forced pattern-forming mechanism that we are studying is qualitatively the same if instead of $f_{m,\alpha}$ we use the smooth sigmoid $g_{\gamma,\nu}$ since the neural field model \eqref{eq::NF-intro} is structurally stable.
	\newline
	\paragraph{\textbf{The intra-neural connectivity parameter $\mu>0$}:} Following our previous works \cite{tamekue2024mathematical,tamekue2023,tamekue2022reproducing}, we assume that $\mu>0$ is smaller than the threshold parameter $\mu_c>0$ where certain geometric patterns spontaneously emerge in V1 in the absence of sensory inputs from the retina, see for instance, \cite{bressloff2001,curtu2004,ermentrout1979,nicks2021}. This threshold parameter is referred to as the bifurcation point, and it is analytically given by
	\begin{equation}\label{eq::parameter mu_c}
		\mu_c := \frac{1}{f_{m,\alpha}'(0)\widehat{\omega}(q_c)} = \frac{1}{\alpha\widehat{\omega}(q_c)}
	\end{equation}
	where $\widehat{\omega}(q_c)$ is defined by \eqref{eq::maximum of omegahat}. Moreover, we let
	\begin{equation}\label{eq::parameter mu_0}
		\mu_0 := \frac{1}{f_{m,\alpha}'(0)\|\omega\|_1} = \frac{1}{\alpha\|\omega\|_1}\le\mu_c
	\end{equation}
	be the largest value of $\mu$ up to which we can ensure the existence and uniqueness of a stationary solution to \eqref{eq::NF-intro} in the space $L^\infty(\R^2)$. We henceforth assume that
	\begin{equation}\label{eq::general asumption}
		\mu<\mu_0.
	\end{equation}
	\begin{remark}
		The response function $f_{m,\alpha}$ is globally bounded for all finite $m\ge 0$ and $\alpha>0$ ensuring that, independently of $\mu>0$, the solution $a\in C([0,\infty);L^\infty(\R^d))$ of \eqref{eq::NF-intro} is uniformly bounded for $t\in [0,+\infty)$, for any initial datum $a_0\in L^\infty(\R^2)$ and sensory input $I\in L^\infty(\R^2)$. See for instance \cite[Theorem~B.6.]{tamekue2024mathematical}.
		Although the semilinear response function $f_{\infty,\alpha}$ is unbounded, we prove in Section~\ref{s::well-posedness} that this is still true under the assumption $\mu<\mu_{0}$.
	\end{remark}

	\subsection{Mathematical modeling of visual stimuli}\label{ss::mathematical modeling}
	In this section, we mathematically model the cortical representation of visual stimuli associated with funnel patterns used in Billock and Tsou's experiments. Then, we incorporate them as sensory inputs in \eqref{eq::NF-intro}. Note that we are devoted to modeling the static version of these phenomena. Here, ``static'' refers to a physical visual stimulus that induces an afterimage on the retina, resulting in illusory contours that do not exhibit apparent motion. Consequently, we will not consider a time-dependent sensory input, which should incorporate the modeling of flickering employed in the experiment. However, as we already pointed out, this consideration will be enough for the corresponding stationary output pattern of \eqref{eq::NF-intro} to capture all the essential features (illusory contours) of the after-image reported by Billock and Tsou.
	
	Recall that the functional architecture of V1 exhibits a remarkable characteristic known as retinotopic organization \cite{tootell1982}: the neurons in the V1 area are arranged orderly, forming a topographic or retinotopic map (well-known as the retino-cortical map). This map represents a two-dimensional projection of the visual image formed on the retina. Notably, neighboring regions of the visual field are represented by neighboring regions of neurons in V1, establishing a bijective relationship.
	
	Up to the authors' knowledge, the retino-cortical map was first represented analytically as a complex logarithmic map in \cite{schwartz1977}. Let $(r,\theta)\in[0,\infty)\times[0,2\pi)$ denote polar coordinates in the visual field (or in the retina) and $(x_1,x_2)\in\R^2$ Cartesian coordinates in V1. The retino-cortical map (see also \cite{tamekue2022reproducing} and references within) is analytically given by 
	\begin{equation}\label{eq::retino-cortical}
		r e^{i\theta}  \mapsto (x_1,x_2):=\left( \log r, \theta \right).
	\end{equation}
	\begin{figure}[t]
		\centering
		% First image with TikZ square contour
		\begin{minipage}{0.45\textwidth}
			\centering
			\begin{tikzpicture}[remember picture]
				\node[anchor=south west,inner sep=0] (image1) at (0,0) {
					\includegraphics[width=0.75\textwidth]{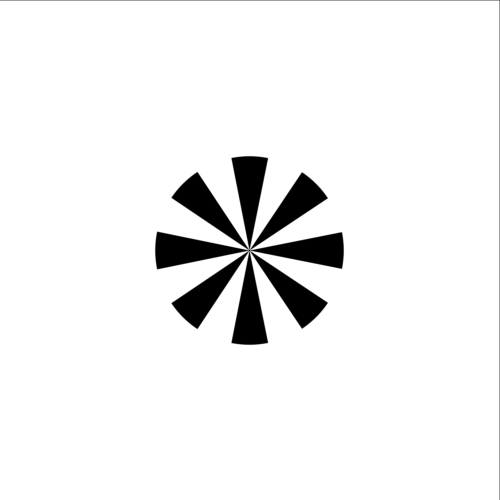}
				};
				\begin{scope}[x={(image1.south east)},y={(image1.north west)}]
					\draw[black, thick, opacity=0.1] (0,0) rectangle (1,1);
				\end{scope}
			\end{tikzpicture}
			\caption{Funnel pattern in the centre of the visual field.}
			\label{fig:fovea funnel}
		\end{minipage}
		\hfill
		% Second image with TikZ square contour
		\begin{minipage}{0.45\textwidth}
			\centering
			\begin{tikzpicture}[remember picture]
				\node[anchor=south west,inner sep=0] (image2) at (0,0) {
					\includegraphics[width=0.75\textwidth]{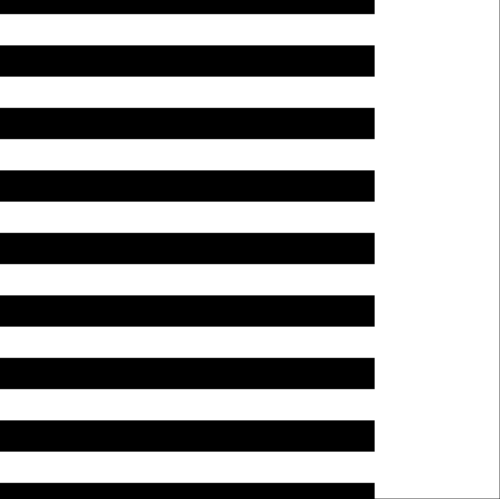}
				};
				\begin{scope}[x={(image2.south east)},y={(image2.north west)}]
					\draw[black, thick, opacity=0.1] (0,0) rectangle (1,1);
				\end{scope}
			\end{tikzpicture}
			\caption{Horizontal stripes in the left area of V1.}
			\label{fig:left horizontal stripes}
		\end{minipage}
		% Arrow with text (above) and math (below), drawn after both images have been placed
		\begin{tikzpicture}[remember picture, overlay]
			\draw[->, thick] ([xshift=5pt]image1.east) -- ([xshift=-5pt]image2.west)
			node[midway, above, font=\small] {retino-cortical map}
			node[midway, below, font=\small] {$r e^{i\theta}  \mapsto\left( \log r, \theta \right)$};
		\end{tikzpicture}
	\end{figure}
	Due to the retino-cortical map analytical representation~\eqref{eq::retino-cortical} and consistent with spontaneous patterns description \cite{bressloff2001,ermentrout1979}, we consider that the function which generates the funnel pattern is given in Cartesian coordinates $x:=(x_1,x_2)\in\R^2$ of V1 by
	\begin{equation}\label{eq::funnel pattern}
		P_F(x) = \cos(2\pi\lambda x_2), \qquad\lambda>0.%\qquad\qquad P_T(x) = \cos(2\pi\lambda x_1),\quad\lambda>0.
	\end{equation}
	
	Let us point out that one of the fundamental properties of the retinotopic projection of the visual field into V1 is that small objects centred on the fovea (centre of the visual field) have a much larger representation in V1 than do similar objects in the peripheral visual field. 
	\begin{figure}[t]
		\centering
		% First image with TikZ square contour
		\begin{minipage}{0.45\textwidth}
			\centering
			\begin{tikzpicture}[remember picture]
				\node[anchor=south west,inner sep=0] (image1) at (0,0) {
					\includegraphics[width=0.75\textwidth]{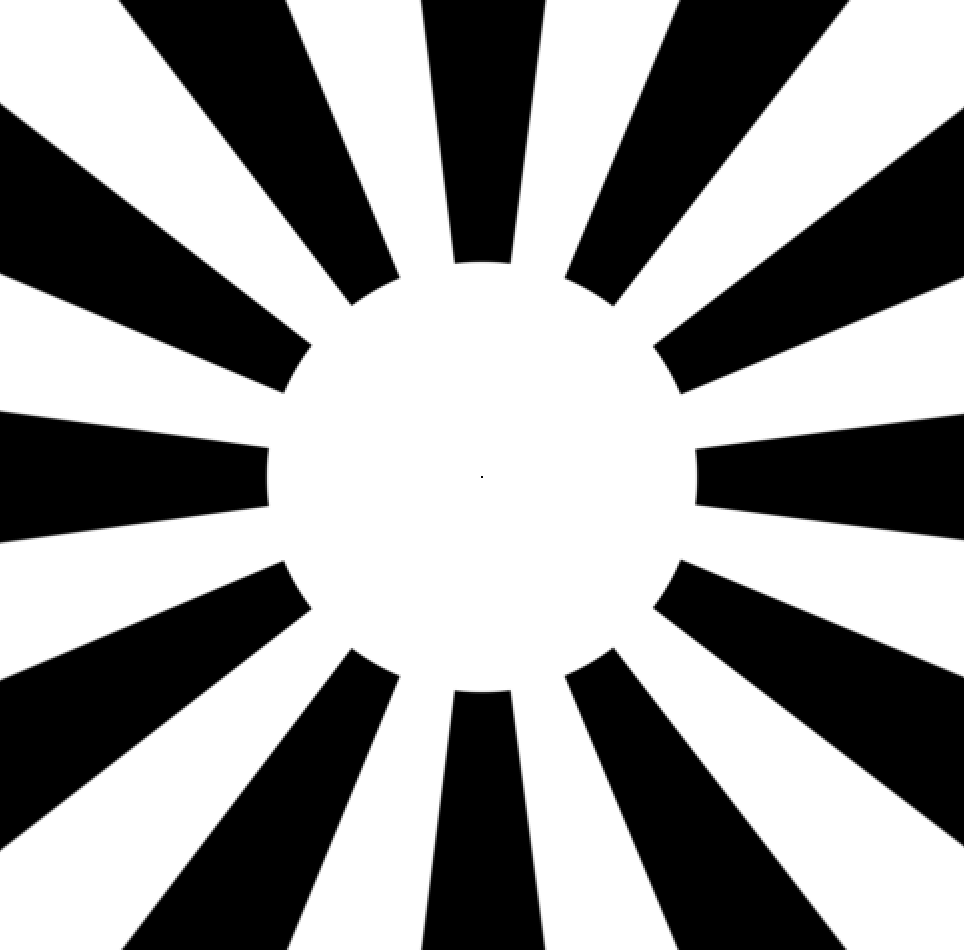}
				};
				\begin{scope}[x={(image1.south east)},y={(image1.north west)}]
					\draw[black, thick, opacity=0.1] (0,0) rectangle (1,1);
				\end{scope}
			\end{tikzpicture}
			\caption{Funnel pattern in the peripheral visual field.}
			\label{fig:peripheral funnel}
		\end{minipage}
		\hfill
		% Second image with TikZ square contour
		\begin{minipage}{0.45\textwidth}
			\centering
			\begin{tikzpicture}[remember picture]
				\node[anchor=south west,inner sep=0] (image2) at (0,0) {
					\includegraphics[width=0.75\textwidth]{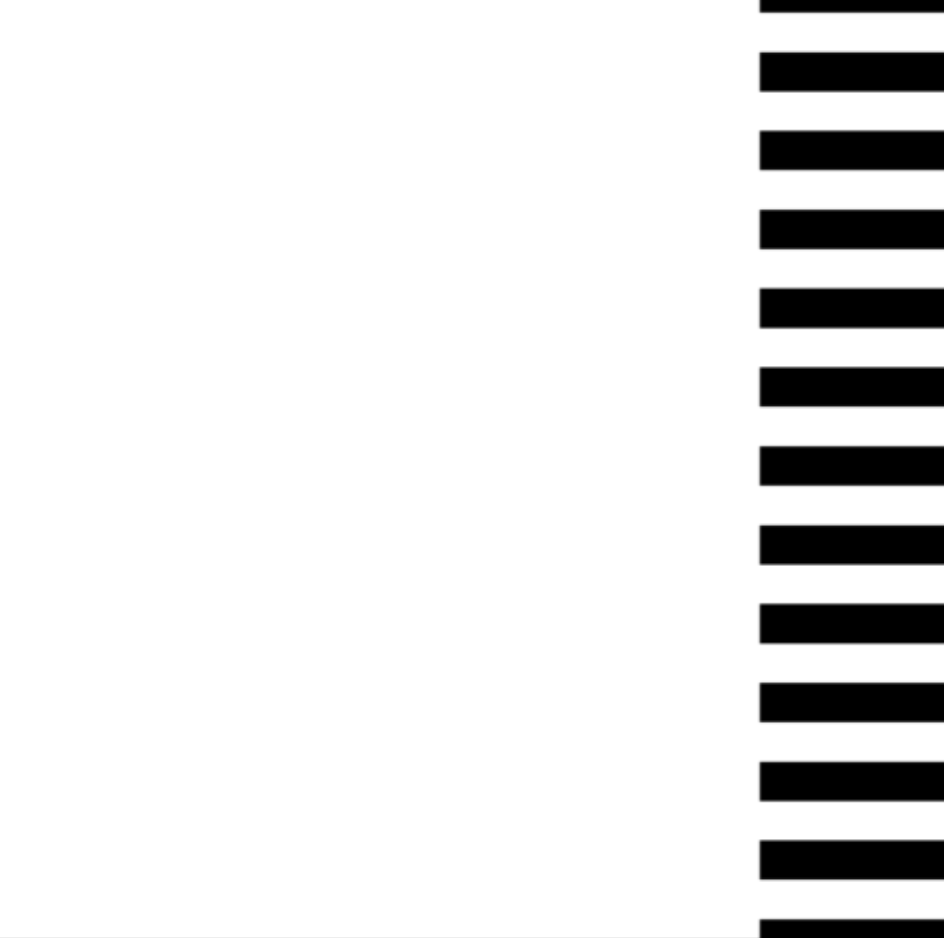}
				};
				\begin{scope}[x={(image2.south east)},y={(image2.north west)}]
					\draw[black, thick, opacity=0.1] (0,0) rectangle (1,1);
				\end{scope}
			\end{tikzpicture}
			\caption{Horizontal stripes in the right area of V1.}
			\label{fig:right horizontal stripes}
		\end{minipage}
		% Arrow with text (above) and math (below), drawn after both images have been placed
		\begin{tikzpicture}[remember picture, overlay]
			\draw[->, thick] ([xshift=5pt]image1.east) -- ([xshift=-5pt]image2.west)
			node[midway, above, font=\small] {retino-cortical map}
			node[midway, below, font=\small] {$r e^{i\theta}  \mapsto\left( \log r, \theta \right)$};
		\end{tikzpicture}
	\end{figure}
	As a result, the cortical representation of Billock and Tsou's visual stimulus associated, e.g., with the funnel pattern localized respectively in the fovea and in the peripheral visual field, should consist of taking the sensory input as 
	\begin{equation}\label{eq::sensory inputs}
		I_L(x_1,x_2) = P_F(x_1,x_2)H(\theta_L-x_1),\qquad I_R(x_1,x_2) = P_F(x_1,x_2)H(x_1-\theta_R).    
	\end{equation}
	Here, $\theta_L$ and $\theta_R$ are nonnegative real numbers, and $H$ is the Heaviside step function, modeling that the funnel pattern is localized in the fovea and the peripheral visual field, respectively. Note that $I_L$ and $I_R$ correspond to sensory inputs consisting of horizontal stripes in the left and right areas of the cortex V1. Indeed, since visual stimuli employed in these experiments are alternating sequences of white and black zones, we represent every cortical function, say $I_R$, as defined in \eqref{eq::sensory inputs} in terms of a binary image, corresponding to the zero-level set of $I_R$, in the following way: in the regions where $I_R>0$ we put the \emph{black grayscale} and where $I_R\le 0$ we put the \emph{white grayscale}, refer for instance, to Figures~\ref{fig:left horizontal stripes} and~\ref{fig:right horizontal stripes}.
	\begin{remark}
		For ease in the presentation, in the following, we will restrict ourselves to the funnel pattern $I_L$ localized in the left area of V1 since the same analysis can be straightforwardly adapted to $I_R$.
	\end{remark}
	
	\section{Preliminary results on the Amari-type equation}\label{s::well-posedness}
	In this section, we begin by discussing the concept of a stationary state as it applies to \eqref{eq::NF-intro}. Following this, we review essential preliminary findings related to the well-posedness of the same equation that is necessary to comprehend the rest of the paper.
	\begin{definition}[Stationary state]\label{def::stationary state to WC equation}
		Let $a_0\in L^p(\R^2)$. For every $I\in L^p(\R^2)$, a stationary state $a_I\in L^p(\R^2)$ to \eqref{eq::NF-intro} is  a time-invariant solution, viz.
		\begin{equation}\label{eq::SS}\tag{SS}
			a_I = \mu\omega\ast f(a_I)+I.
		\end{equation}
	\end{definition}
	
	The following well-posedness result is \cite[Theorem~3.1]{tamekue2024mathematical}, which only relies on the global Lipschitz property of the nonlinearity $f$. 
	\begin{theorem}[\cite{tamekue2024mathematical}]\label{thm::existence of stationary input}
		Let $I\in L^\infty(\R^2)$. For any initial datum $a_0\in L^\infty\R^2)$, there exists a unique $a\in C([0,\infty);L^\infty(\R^d))$, solution to \eqref{eq::NF-intro}. If $\mu<\mu_0$, there exists a unique stationary state $a_I\in L^\infty(\R^2)$ to \eqref{eq::NF-intro}. Moreover, the following holds.
		\begin{equation}\label{eq::Decay of solution}
			\|a(\cdot,t)-a_I\|_{\infty}\le e^{-(1-\mu \|\omega\|_1)t}\|a_0-a_I\|_\infty,
			{\qquad \text{for any } t\ge 0.}
		\end{equation}
	\end{theorem}
	In the following theorem, we prove the uniform boundedness of the solution under the assumptions of Section~\ref{ss::Assumption on parameters}. 
	\begin{theorem}\label{thm::boundness of the solution}
		Let $a_0\in L^\infty(\R^2)$, $I\in L^\infty(\R^2)$ and $a\in C([0,\infty);L^\infty(\R^d))$ be the solution of \eqref{eq::NF-intro}. Then,  
		\begin{itemize}
			\item[i.] If $0<\mu<\mu_0$, it holds
			\begin{equation}\label{eq::bound on limsup of a mu less than mu_0}							\|a_I\|_\infty = \lim\limits_{t\rightarrow+\infty}\|a(\cdot,t)\|_\infty\le\|I\|_\infty\left(1-\frac{\mu}{\mu_0}\right)^{-1}
			\end{equation}
			where $a_I$ is the stationary solution to Equation \eqref{eq::NF-intro} given by Theorem~\ref{thm::existence of stationary input}.
			\item[ii.] If $\mu=\mu_0$, we have
			\begin{equation}\label{eq::bound on limsup of a mu tends to mu_0}
				\|a(\cdot,t)\|_\infty\le \|I\|_\infty t+\|a_0\|_\infty, \qquad \text{for any } t\ge 0.
			\end{equation}
		\end{itemize}
	\end{theorem}
	\begin{proof}
		We recall from Theorem~\ref{thm::existence of stationary input} that
		for all $x\in\R^2$, and every $t\ge 0$ we have 
		\begin{equation}\label{eq::Variation of constant formula NF}
			a(x,t) = e^{- t}a_0(x)+\left(1-e^{- t}\right)I(x)+\mu\int_{0}^{t}e^{-(t-s)}(\omega\ast f(a))(x,s)ds.
		\end{equation}
		Therefore, we apply Minkowski's and Young convolution inequalities to \eqref{eq::Variation of constant formula NF}, and obtain for any $t\ge 0$,
		\begin{equation}\label{eq::bound on absolute value of a}
			e^{t}\|a(\cdot,t)\|_\infty\le\|a_0\|_\infty+\frac{\mu}{\mu_0}\int_{0}^{t}e^{s}\|a(\cdot,s)\|_\infty ds+\|I\|_\infty\int_{0}^{t}e^{s}ds
		\end{equation}
		using that $f$ is $\alpha$-Lipschitz continuous. Applying Grönwall's Lemma~\ref{lem::gronwall lemma}
		with $u(t) = e^{t}\|a(\cdot,t)\|_\infty$, $g(t) = \mu/\mu_0$ and $h(t)=\|I\|_\infty e^{t}$ to \eqref{eq::bound on absolute value of a} yields \eqref{eq::bound on limsup of a mu tends to mu_0} for $\mu=\mu_0$, while for $\mu\neq\mu_0$ one gets
		\begin{equation}\label{eq::bound on transient sol a}
			\|a(\cdot,t)\|_\infty\le e^{-\left(1-\frac{\mu}{\mu_0}\right)t}\|a_0\|_\infty+\|I\|_\infty\left(1-\frac{\mu}{\mu_0}\right)^{-1}\left(1-e^{-\left(1-\frac{\mu}{\mu_0}\right)t}\right), \qquad \forall t\ge 0.
		\end{equation}
		Inequality  \eqref{eq::bound on limsup of a mu less than mu_0} follows directly.
	\end{proof}
	One also has the following.
	\begin{proposition}\label{pro::BT NC}
		Under the assumption $\mu<\mu_0$, for any $\alpha>0$ we let 
		\begin{equation}
			m_{\alpha}:=\alpha \|I\|_\infty \left(1-\frac{\mu}{\mu_0}\right)^{-1}.
		\end{equation}
		Then, for any $m\ge m_{\alpha}$
		the stationary solution of \eqref{eq::NF-intro} with response function $f_{m,\alpha}$ coincides with the unique stationary solution to the same equation with response function $f_{m_\alpha,\alpha}$.
	\end{proposition}
	\begin{proof}
		By Theorem~\ref{thm::existence of stationary input}, the stationary solution $a_{m,\alpha}\in L^\infty(\R^2)$ to \eqref{eq::NF-intro} with response function $f_{m,\alpha}$ is the unique solution of $a_{m,\alpha} = I+\mu\omega\ast f_{m,\alpha}(a_{m,\alpha})$. In particular, { by definition of $m_\alpha$, } inequality~\eqref{eq::bound on limsup of a mu less than mu_0} implies that
		\begin{equation}
			-\frac{m_\alpha}{\alpha}\le a_{m,\alpha}(x)\le \frac{m_\alpha}{\alpha}, \qquad \text{for a.e. } x\in\R^2.
		\end{equation}
		Therefore, one has\footnote{Here $\mathbbm{1}_A$ denotes the characteristic function of the subset $A\subset\R^2$.}, for a.e. $x\in\R^2$,
		\begin{eqnarray}
			a_{m,\alpha}(x) &=& I(x)+\mu\omega\ast f_{m,\alpha}(a_{m,\alpha})(x)\nonumber\\
			&=&I(x)+\mu\int_{\R^2}\omega(x-y)f_{m,\alpha}(a_{m,\alpha}(y))\mathbbm{1}_{\{-\frac{m_\alpha}{\alpha}\le a_{m,\alpha}(y)\le\frac{1}{\alpha}\}}dy\nonumber\\
			&&+\mu\int_{\R^2}\omega(x-y)f_{m,\alpha}(a_{m,\alpha}(y))\mathbbm{1}_{\{a_{m,\alpha}(y)\ge\frac{1}{\alpha}\}}dy\nonumber\\
			&=& I(x)+\mu\omega\ast f_{m_\alpha,\alpha}(a_{m,\alpha})(x)
		\end{eqnarray}
		since $f_{m,\alpha}(s)=f_{m_\alpha,\alpha}(s)$ for every $s\ge -m_\alpha/\alpha$. It follows that $a_{m,\alpha}$ is a stationary solution for \eqref{eq::NF-intro} with nonlinearity $f_{m_\alpha,\alpha}$. The statement follows by the uniqueness of the stationary solution provided by Theorem~\ref{thm::existence of stationary input}.
	\end{proof}
	Applied, for instance, to Billock and Tsou's experiments modeling, Proposition~\ref{pro::BT NC} implies the following simple but important result.
	\begin{corollary}\label{coro::BT NC}
		Under the same assumptions as Proposition~\ref{pro::BT NC}, let  $\overline{m}\ge m_\alpha$ be such that the response function $f_{\overline{m},\alpha}$ reproduces Billock and Tsou's experiments. Then, the same is true for any response function $f_{m,\alpha}$ such that $m\ge \overline{m} $.
	\end{corollary}
	
	The following result proves that the stationary state to \eqref{eq::NF-intro} is Lipschitz continuous whenever the sensory input $I$ is.
	\begin{proposition}\label{pro::Lipschitz continuity}
		Assume that $\mu<\mu_0$. If the sensory input $I\in L^\infty(\R^2)$ is $L_I$-Lipschitz continuous on some open set $\Omega\subset\R^2$, then the corresponding stationary solution to equation \eqref{eq::NF-intro} is also Lipschitz continuous on $\Omega$, with Lipschitz constant upper bounded by
		\begin{equation}
			\label{eq:DI}
			D_I:= L_I + \mu m_{\alpha}C_\omega 
		\end{equation}
		where $C_\omega$ denotes a constant depending only on $\omega$.
	\end{proposition}
	\begin{proof}
		% Let $I\in L^\infty(\R^2)$, then 
		Let $a\in L^\infty(\R^2)$ be the unique stationary solution whose existence is guaranteed by Theorem~\ref{thm::existence of stationary input}.
		For $x\in\R^2$ we have that
		\begin{equation}
			a(x) = I(x)+\mu b(x)\quad\mbox{with}\quad b := \omega\ast f(a).
		\end{equation}
		Since $\omega\in \cS(\R^2)$ and $f(a)\in L^\infty(\R^2)$, one has that $b$ is infinitely differentiable on $\R^2$. 
		Since by assumption $f$ is $\alpha$-Lipschitz continuous and satisfies $f(0)=0$, it is straightforward to show that
		\begin{equation}
			\|\nabla b(x)\|\le \alpha \|a\|_\infty\sqrt{\|\partial_{x_1}\omega\|_1^2+\|\partial_{x_2}\omega\|_1^2},\qquad\forall x\in\R^2.
		\end{equation}
		It follows by the Mean Value Theorem that $b$ is Lipschitz continuous on $\R^2$. Since $I$ is Lipschitz continuous on $\Omega$ and using Theorem \ref{thm::boundness of the solution} to upper bound $\|a\|_\infty$, the result then follows at once.
	\end{proof}
	
	The following simple result will be used hereafter.
	\begin{lemma}\label{lem::general when f is odd}
		Assume that the response function $f$ in \eqref{eq::NF-intro} is odd.
		If $\mu<\mu_0$, for any sensory input $I\in L^\infty(\R^2)$ one has $a_{ \shortminus I}=-a_I$.
	\end{lemma}
	\begin{proof}
		Thanks to Theorem~\ref{thm::existence of stationary input}, we have that $a_I$ and $a_{\shortminus I}$ are uniquely defined by $a_I =I+\mu\omega\ast f(a_I)$ and $a_{ \shortminus I} =-I+\mu\omega\ast f(a_{ \shortminus I})$, respectively. Since $f$ is odd, one has $a_I+a_{ \shortminus I} = \mu\omega\ast[f(a_I)-f(-a_{ \shortminus I})]$. Therefore, $\|a_I+a_{ \shortminus I}\|_\infty=0$ since Young convolution inequality gives 
		\[
		\|a_I+a_{ \shortminus I}\|_\infty\le\mu\|\omega\|_1\|f(a_I)-f(-a_{ \shortminus I})\|_\infty\le\frac{\mu}{\mu_0}\|a_I+a_{ \shortminus I}\|_\infty.\qedhere
		\] 
	\end{proof}
	
	In the following, we prove more general results that provide insight into the qualitative properties of the stationary state of \eqref{eq::NF-intro} when the sensory input has a cosine factor.
	% will allow us to solve, in particular, the conjecture announced in \cite[Conjecture~1]{tamekue2022reproducing}.
	\begin{proposition}\label{pro::general when f is odd}
		Let the sensory input $I$ be given by $I(x_1,x_2) = \cos(2\pi\lambda x_2)I_1(x_1)$, for $\lambda>0$ and $(x_1,x_2)\in\R^2$, where $I_1\in L^\infty(\R)$. If  $\mu<\mu_0$, then the following hold.
		\begin{enumerate}
			\item $a_I$ is $1/\lambda$-periodic, even and globally Lipschitz continuous with respect to $x_2\in\R$;
			\vspace{.4em}
			\item If $f$ is odd, then $a_I$ is $1/2\lambda$-antiperiodic with respect to $x_2\in\R$. Namely,
			\begin{equation}
				a_I(x_1,x_2+1/2\lambda) = -a_I(x_1,x_2),\qquad\mbox{for a.e. }\; (x_1,x_2)\in\R^2.
			\end{equation}
		\end{enumerate}
	\end{proposition}
	\begin{proof}
		We assume that $\lambda=1$ for ease of notation. Using that the convolution operator commutes with translation and that the input $I$ and the kernel $\omega$ are even with respect to $x_2$, one deduces that $a_I$ is even with respect to $x_2$.  
		Let us prove that $a_I$ is $1$-periodic with respect to $x_2$. For a.e.~$(x_1,x_2)\in \R^2$, one has
		\begin{eqnarray}
			a_I(x_1,x_2+1)&=&\cos(2\pi\lambda x_2)I_1(x_1)+\mu\int_{\R^2}\omega(x_1-y_1,x_2+1-y_2)f(a_I(y))dy\nonumber\\
			&=& I(x_1,x_2)+\mu\int_{\R^2}\omega(x-y)f(a_I(y_1,y_2+1))dy.
		\end{eqnarray}
		It follows that $(x_1,x_2)\mapsto a_I(x_1,x_2+1)$ is the stationary solution associated with $I$ and hence it coincides with $a_I$.
		
		Let us show that $a_I$ is Lipschitz continuous with respect to $x_2$. Taking the derivative of \eqref{eq::level set 1} with respect to $x_2$, one finds that for a.e.~$(x_1,x_2)\in\R^2$ it holds
		\begin{equation}\label{eq::level set 2}
			\partial_{x_2}a(x_1,x_2) = -2\pi\sin(2\pi x_2)I_1(x_1)+\mu\int_{\R^2}\omega(x-y)f'(a_I(y))\partial_{x_2}a(y)dy.
		\end{equation}
		Since $\|f'\|_\infty\le\alpha$ by assumption, it follows that
		\[
		\|\partial_{x_2}a(x_1,\cdot)\|_{L^\infty(\R)}\le 2\pi\|I_1\|_\infty\left(1-\mu/\mu_0\right)^{-1}
		\] 
		for a.e. $x_1\in\R$,
		showing that $a_I(x_1,\cdot)$ is Lipschitz continuous for a.e. $x_1\in\R$. This completes the proof of \textit{item}~(1). 
		
		Let us now prove \textit{item}~(2). For a.e.~$(x_1,x_2)\in\R^2$, one has
		\begin{eqnarray}
			-a_I(x_1,x_2+1/2)
			&=&\cos(2\pi\lambda x_2)I_1(x_1)-\mu\int_{\R^2}\omega(x_1-y_1,x_2+1/2-y_2)f(a_I(y))dy\nonumber\\
			&=&I(x_1,x_2)+\mu\int_{\R^2}\omega(x-y)f(-a_I(y_1,y_2+1/2))dy
		\end{eqnarray}
		where in the last equality we used the fact that $f$ is odd. Hence, $(x_1,x_2)\mapsto -a_I(x_1,x_2+1/2)$ is the stationary solution associated with $I$ and hence it coincides with $a_I$.
	\end{proof}

	One has the following result related to Billock and Tsou's experiments, which is the main focus of this paper. 
	The proof is an adaptation of that of \cite[Theorem 5.2]{tamekue2024mathematical}. We will present it for the sake of completeness.
	
	\begin{proposition}\label{pro::general when f bounded and odd BT}
		Assume that the response function $f$ in \eqref{eq::NF-intro} is odd.
		Let the sensory input $I_L$ be given by \eqref{eq::sensory inputs}. If  $\mu<\mu_0/2$, denote by $a_L\in L^\infty(\R^2)$ the corresponding stationary state to $I_L$. Then, for a.e.~$x_1\in\R$, the set of zeros of $a_L(x_1,\cdot)$ coincides with that of $x_2\mapsto\cos(2\pi\lambda x_2)$.
	\end{proposition}
	
	\begin{proof}
		We assume that $\lambda=1$ for ease of notation. The zeroes of $x_2\mapsto\cos(2\pi x_2)$ are $z_k=k+1/4$ for every $k\in\Z$. Let $x_1\in\R$, let us first prove that $a_I(x_1,z_k) = 0$. Since $a_I(x_1,\cdot)$ is $1$-periodic by Proposition~\ref{pro::general when f is odd}, it is enough to prove that $a_I(x_1,1/4) = 0$. Using that $a_I(x_1,\cdot)$ is $1/2$-antiperiodic and even by Proposition~\ref{pro::general when f is odd}, one gets $a_I(x_1,1/4) = a_I(x_1,-1/4+1/2)=-a_I(x_1,-1/4)=-a_I(x_1,1/4)$. Therefore, $a_I(x_1,1/4) = 0$. Conversely, let $x^{*}:=(x_1^{*},x_2^{*})\in\R^2$ be such that $a_L(x^{*}) = 0$. We want to show that $\cos(2\pi x_2^{*})=0$. Recall that for a.e. $x:=(x_1,x_2)\in\R^2$,
		\begin{equation}\label{eq::aL}
			a_L(x) = \cos(2\pi x_2)H(\theta_L-x_1)+\mu\int_{\R^2}\omega(x-y)f(a_L(y))dy.
		\end{equation}
		If $x_1^{*}\le\theta_L$, then from \eqref{eq::aL}, it follows
		\begin{equation}\label{eq::useful-1}
			\cos(2\pi x_2^*) = -\mu\int_{\R^2}\omega(y)f(a_F(x_*-y))dy.
		\end{equation}
		By using \eqref{eq::aL} once again, one obtains 
		\begin{equation}\label{eq::Psi intermediate 1}
			a_L(x_1^{*}-y_1,x_2^{*}-y_2) = I_2(y)+\mu\int_{\R^2}k(y,z)f(a_L(x^{*}-z))dz
		\end{equation}
		where $I_2(y):=\sin(2\pi x_2^{*})\sin(2\pi y_2)H(\theta_L-x_1^{*}+y_1)$ and for every $(x,y)\in\R^2\times\R^2$
		\begin{equation}\label{eq::kernel k}
			k(y,z) = \omega(y-z)-H(\theta_L-x_1^{*}+y_1)\cos(2\pi y_2)\omega(z)
		\end{equation}
		satisfies
		\begin{equation}
			\sup\limits_{y\in\R^2}\int_{\R^2}|k(y,z)|dy\le 2\|\omega\|_1.
		\end{equation}
		Since $\mu<\mu_0/2$, the contracting mapping principle ensures that $y\mapsto g_2(y):=~a_L(x_*-y)$ is the unique solution to \eqref{eq::Psi intermediate 1}. Moreover, it holds
		$$
		-a_L(x_1^{*}-y_1,x_2^{*}+y_2) =  I_2(y)+\mu\int_{\R^2}k(y,z)f(-a_L(x_1^{*}-z_1,x_2^{*}+z_2))dz
		$$
		since $f$ is odd. So the function $(y_1,y_2)\mapsto -g_2(y_1,-y_2)$ is also solution of \eqref{eq::Psi intermediate 1} with input $I_2$. By uniqueness of solution, one has $g_2(y_1,-y_2) = -g_2(y_1,y_2)$ and that $y\mapsto\omega(y)f(a_L(x_*-y))$ is an odd function with respect to $y_2\in\R$, since $\omega$ is radial and $f$ is an odd function. It follows from Fubini's theorem that the right-hand side of \eqref{eq::useful-1} is equal to $0$.
	\end{proof}
	
	\begin{remark}
		Note that the assumption $\mu<\mu_0/2$ in Proposition~\ref{pro::general when f bounded and odd BT} is technical due to our strategy in the proof. Numerical simulations suggest that the conclusion of the proposition remains valid for all $\mu_0/2\le\mu<\mu_0$. See, for instance, Figure~\ref{fig:left horizontal stripes odd sigmoid}.
	\end{remark}
	
	\begin{figure}%[ht]
		\centering
		% First image with TikZ square contour
		\begin{minipage}{0.45\textwidth}
			\centering
			\begin{tikzpicture}[remember picture]
				\node[anchor=south west,inner sep=0] (image1) at (0,0) {
					\includegraphics[width=0.8\textwidth]{BT-stimulus_fovea-V1.png}
				};
				\begin{scope}[x={(image1.south east)},y={(image1.north west)}]
					\draw[black, thick, opacity=0.1] (0,0) rectangle (1,1);
				\end{scope}
			\end{tikzpicture}
			% \caption{Funnel pattern in the centre of the visual field.}
			% \label{fig:fovea funnel odd sigmoid}
		\end{minipage}
		\hfill
		% Second image with TikZ square contour
		\begin{minipage}{0.45\textwidth}
			\centering
			\begin{tikzpicture}[remember picture]
				\node[anchor=south west,inner sep=0] (image2) at (0,0) {
					\includegraphics[width=0.8\textwidth]{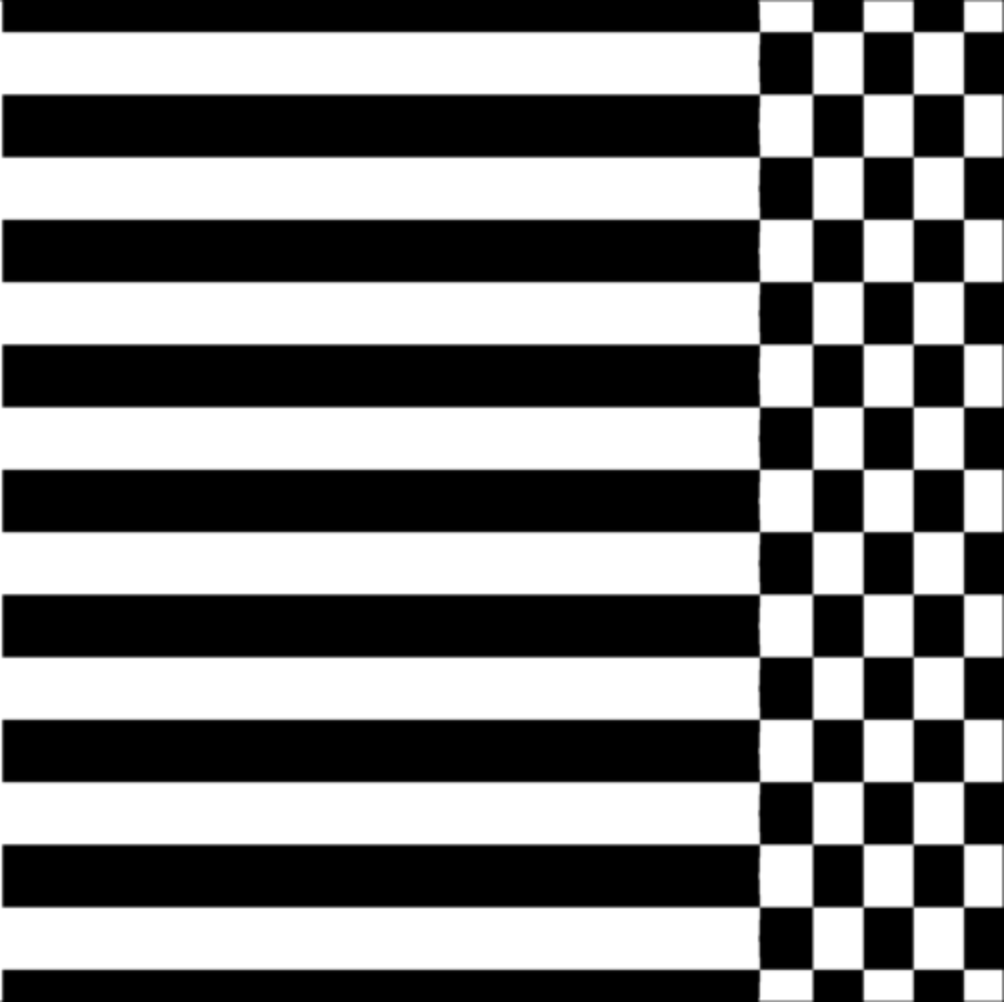}
				};
				\begin{scope}[x={(image2.south east)},y={(image2.north west)}]
					\draw[black, thick, opacity=0.1] (0,0) rectangle (1,1);
				\end{scope}
			\end{tikzpicture}
			% \caption{Horizontal stripes in the left area of V1.}
			% \label{fig:left horizontal stripes odd sigmoid}
		\end{minipage}
		% Arrow with text (above) and math (below), drawn after both images have been placed
		\begin{tikzpicture}[remember picture, overlay]
			\draw[->, thick] ([xshift=5pt]image1.east) -- ([xshift=-5pt]image2.west)
			node[midway, above, font=\small] {via the stationary}
			node[midway, below, font=\small] {equation~\eqref{eq::SS}};
		\end{tikzpicture}
		\caption{On the \textit{left}, we have the sensory input $I_L(x_1, x_2) = \cos(2\pi\lambda x_2)H(\theta_L-x_1)$ with $\lambda=0.4$ and $\theta_L=5$. On the \textit{right}, we have the corresponding stationary output when the response function is the odd nonlinearity $f_{1,1}(s)=\max(-1,\min(1,s))$. The cortical data is defined on the square $(x_1,x_2)\in [-10, 10]^2$ with step $\Delta x_1=\Delta x_2=0.01$. The parameters in the kernel $\omega$ are $\sigma_1=1/\pi$, $\sigma_2=\sqrt{2}/\pi$ and $\kappa=1.2$. Here $\mu:=0.99\mu_0$, where $\mu_0$ is defined in \eqref{eq::parameter mu_0}-\eqref{eq::L^1-norm of omega}. These numerical results are obtained using the Julia toolbox from \cite{tamekue:tel-04230895}.}
		\label{fig:left horizontal stripes odd sigmoid}
	\end{figure}

	\section{On Billock and Tsou's experiments modelling}\label{s::theoretical Billock and Tsou}
	In this section, we investigate the modeling of Billock and Tsou's phenomena using \eqref{eq::NF-intro}. In the current study, we aim to elucidate the efficacy of \eqref{eq::NF-intro} in mimicking these visual illusions, as briefly reviewed in Section~\ref{ss::BT experiments}. We focus on determining if the model's output exhibits qualitative concordance with the human experience of these illusions. It is imperative to note that our analysis is mechanistic and strictly qualitative and serves as an illustrative proof of concept for applying Amary-type dynamics \eqref{eq::NF-intro} in simulating the perceptual outcomes elicited by visual illusions as previously obtained in \cite{tamekue2024mathematical} for the visual MacKay effect modeling. This study does not endeavor to align its findings with quantitative empirical data, as such data are contingent upon numerous experimental variables that do not have a straightforward relationship with the parameters within our model.

	We begin by proving that these phenomena are wholly nonlinear in contrast, for instance, to the visual MacKay effect \cite{mackay1957} that we proved in \cite{tamekue2024mathematical} for being a linear phenomenon. Therefore, we will see that \eqref{eq::NF-intro} with a linear response function $f$ cannot reproduce the psychophysical experiments by Billock and Tsou \cite{billock2007} associated with the funnel pattern stimulus when the corresponding sensory inputs are modeled as in \eqref{eq::sensory inputs}. 
	
	\subsection{Unreproducibility of Billock and Tsou experiments: linear response function}\label{ss::theoretical Billock and Tsou experiments linear}
	This section assumes that the response function $f$ is linear. To simplify our analysis, we specifically focus on the funnel pattern centred on the fovea within the visual field. As a result, the corresponding sensory input $I$ consists of a localized pattern of horizontal stripes in the left area of the V1 cortex by the retino-cortical map, see Figure~\ref{fig:left horizontal stripes}.
	
	Previously, in \cite[Proposition 5.]{tamekue2023}, we proved that \eqref{eq::NF-intro} with a linear response function is incapable of reproducing Billock and Tsou's experiments, as verified through direct Fourier transform computations. While this finding sufficed to establish our desired outcome, it failed to offer deeper insights into the qualitative properties of the stationary state associated with the sensory input utilized in these experiments. Specifically, it did not precisely characterize the zero-level set of this stationary state. To address this gap, we draw upon the qualitative properties of the sensory input $I$ and utilize tools from complex and harmonic analysis.  Consequently, we present the following key results in this section.
	\begin{theorem}\label{thm::zero-level set linear BT}
		Assume that the response function $f$ in \eqref{eq::NF-intro} is linear with slope $\alpha>0$ and that the sensory input $I=I_L$.  If  $\mu<\mu_0$, denote by $a_L\in L^\infty(\R^2)$ the corresponding stationary state to $I_L$. Then, the zero-level set $\cZ_{a_L}$ of $a_L$ satisfies
		\begin{equation}\label{eq::zero-level set linear BT}
			\cZ_{a_L}\cap[(0,+\infty)\times\R] = [\cX_1\times\R]\cup[(0,+\infty)\times\cX_2]
		\end{equation}
		where $\cX_1$ and $\cX_2$ are discrete and countable sets, respectively in $(0,+\infty)$ and $\R$.
	\end{theorem}
	Since $f(s) = \alpha s$ and $I_L(x_1,x_2) = \cos(2\pi\lambda x_2)H(\theta_L-x_1)$, with $\lambda>0$ and $\theta_L\ge 0$, we assume without loss of generality that $\alpha=1$, $\lambda=1$ and $\theta_L=0$ to keep the presentation clear for reader convenience. Therefore, the stationary state $a_L\in L^\infty(\R^2)$ satisfies
	\begin{equation}\label{eq::level set 1}
		a_L(x_1,x_2) = \cos(2\pi x_2)H(-x_1)+\mu\int_{\R^2}\omega(x-y)a_L(y)dy,\qquad (x_1,x_2)\in\R^2
	\end{equation}
	where the kernel $\omega$ is defined in \eqref{eq::connectivity}.
	
	We pedagogically split the proof of Theorem~\ref{thm::zero-level set linear BT} into several steps. The first result is the following.
	\begin{lemma}\label{lem::zero-level set linear BT 2}
		Under hypotheses of Theorem~\ref{thm::zero-level set linear BT}, the stationary state $a_L$ decomposes as 
		\begin{equation}\label{eq::a in lineat regime BT}
			a_L(x_1,x_2) = a_1(x_1)\cos(2\pi x_2),\qquad (x_1,x_2)\in\R^2.
		\end{equation}
		Here $a_1\in L^\infty(\R)$ is given by
		\begin{equation}\label{eq::first Fourier coeffs of a linear}
			a_1(x_1) = H(-x_1)+\mu(K\ast H(-\cdot))(x_1),\qquad x_1\in\R.
		\end{equation}
		where $K\in\cS(\R)$ is defined for all $x_1\in\R$ by
		\begin{equation}\label{eq::noyau K}
			K(x_1) = \int_{-\infty}^{+\infty}e^{2i\pi x_1\xi}\frac{\widehat{\psi_1}(\xi)}{1-\mu\widehat{\psi_1}(\xi)}d\xi,\quad \widehat{\psi_1}(\xi) = e^{-2\pi^2\sigma_1^2(1+\xi^2)}-\kappa e^{-2\pi^2\sigma_2^2(1+\xi^2)}
		\end{equation}
		for every $ \xi\in\R$.
	\end{lemma}
	\begin{proof}
		We fix $x_1\in\R$. Since $x_2\mapsto a_L(x_1,x_2)$ is $1$-periodic and even on $\R$, we expand $a_L(x_1,\cdot)$ in term of Fourier series as
		\begin{equation}\label{eq::linear Fourier series}
			a_L(x_1,x_2) = \sum\limits_{n=0}^{\infty}a_n(x_1)\cos(2\pi n x_2),\qquad\qquad  x_2\in\R
		\end{equation}
		\begin{equation}\label{eq::linear Fourier coefficients}
			a_0(x_1) = \int_{0}^{1}a_L(x_1,t)dt\qquad\mbox{and}\qquad a_n(x_1) = 2\int_{0}^{1}a_L(x_1,t)\cos(2\pi nt)dt.
		\end{equation}
		Thanks to the \textit{item}~(1) of Proposition~\ref{pro::general when f is odd}, one has that the derivative $a_L'(x_1,\cdot)$ of $a_L$ with respect to $x_2$ is continuous and bounded on $\R$. Thus $a_L'(x_1,\cdot)$ belongs to $L^2([-1,1])$, the space of real-valued measurable and square-integrable functions over $[-1,1]$. Since $a_L(x_1,\cdot)$ is absolutely continuous (Lipschtiz continuous still by the \textit{item}~(1) of Proposition~\ref{pro::general when f is odd}) on $\R$, it follows from \cite[Théorème 2.]{kolmogorov1974} that its Fourier series converges uniformly to $a_L(x_1,\cdot)$ on $\R$.
		Observe also that \eqref{eq::linear Fourier coefficients} defines functions $a_n\in L^\infty(\R)$ for all $n\in\N$, so that one gets for all $x_1\in\R$ and for all $\sigma>0$, the existence of $M>0$ such that
		\begin{equation}
			\sum\limits_{n=0}^{+\infty}\int_{-\infty}^{\infty}\left|\frac{1}{\sigma\sqrt{2\pi}}e^{-\frac{(x_1-y_1)^2}{2\sigma^2}}e^{-2\pi ^2\sigma^2n^2}a_n(y_1)\cos(2\pi nx_2)\right|dy_1 \le \frac{M}{1-e^{-2\pi ^2\sigma^2}}.
		\end{equation}
		Therefore, we can substitute \eqref{eq::linear Fourier series} into \eqref{eq::level set 1} and find the following family of one-dimensional linear integral equations indexed by $n\in\N$.
		\begin{equation}\label{eq::Fourier coefficients of a linear dim 1}
			a_n(x_1) = \delta_{1,n}H(-x_1)+\mu\alpha(\psi_n\ast a_n)(x_1),\qquad x_1\in\R
		\end{equation}
		where $\delta_{1,n}$ is the usual Kronecker symbol and the kernel $\psi_n$ is given for $n\in\N$, by
		\begin{equation}\label{eq::omega n}
			\psi_n(s) = e^{-2\pi^2n^2\sigma_1^2}\frac{e^{-\frac{s^2}{2\sigma_1^2}}}{\sigma_1\sqrt{2\pi}}-\kappa e^{-2\pi^2n^2\sigma_2^2}\frac{e^{-\frac{s^2}{2\sigma_2^2}}}{\sigma_2\sqrt{2\pi}},\qquad s\in\R.
		\end{equation}
		For $n\neq 1$, equations \eqref{eq::Fourier coefficients of a linear dim 1} yields to 
		\begin{equation}\label{eq::coeffs equation n different de 1}
			(\delta-\mu\alpha\psi_n)\ast a_n = 0,\qquad\mbox{in}\qquad \cS'(\R)
		\end{equation}
		where $\delta$ is the Dirac distribution at $0$. Taking the Fourier transform of \eqref{eq::coeffs equation n different de 1} in the space $\cS'(\R)$, one obtains for all $\xi\in\R$,
		\begin{equation}
			(1-\mu\widehat{\psi_n}(\xi))\cF\{a_n\}(\xi) = 0,\qquad n\neq 1.
		\end{equation}
		It is not difficult to see that $\max\{\widehat{\psi_n}(\xi)\mid\xi\in\R\}\le\max\{\widehat{\omega}(\xi)\mid\xi\in\R^2\}\le \|\omega\|_1$. Since $\mu\|\omega\|_1<1$ by assumption, one deduces $1-\mu\widehat{\psi_n}(\xi)>0$ for all $\xi\in\R$, and $\cF\{a_n\}\equiv 0$. It follows that 
		\begin{equation}
			a_n\equiv 0,\qquad \mbox{for all }\qquad\quad n\neq 1.
		\end{equation}
		In the case $n=1$, one has
		\begin{equation}\label{eq::harmonic 1 of a linear dim 1}
			a_1(x_1) = H(-x_1)+\mu(\psi_1\ast a_1)(x_1),\qquad\qquad x_1\in\R.
		\end{equation}
		Finally, taking respectively the Fourier transform of \eqref{eq::harmonic 1 of a linear dim 1} and the inverse Fourier transform in the space $\cS'(\R)$, we find that $a_1\in L^\infty(\R^2)$ is given by \eqref{eq::first Fourier coeffs of a linear} with $K\in\cS(\R)$ defined as in \eqref{eq::noyau K}.
	\end{proof}
	
	Due to Lemma~\ref{lem::zero-level set linear BT 2}, inverting the kernel $K$ defined in \eqref{eq::noyau K} and providing an asymptotic behavior of its zeroes on $\R$ will help to provide thorough information on the zeroes of the function $a_1$ as given by \eqref{eq::first Fourier coeffs of a linear}. To achieve this, we use tools from complex analysis. 
	
	Let us consider the extension of $\widehat{K}$ in the set $\C$ of complex numbers,
	\begin{equation}
		\widehat{K}(z) = \frac{\widehat{\psi_1}(z)}{1-\mu\widehat{\psi_1}(z)},\qquad z\in\C.
	\end{equation}
	Then $\widehat{K}$ is a meromorphic function on $\C$, and its poles are zeroes of the entire function
	\begin{equation}\label{eq::exponential polynomial h}
		h(z):= 1-\mu e^{-2\pi^2\sigma_1^2(1+z^2)}+\kappa\mu e^{-2\pi^2\sigma_2^2(1+z^2)},\qquad z\in\C.
	\end{equation}
	\begin{remark}\label{rmk::exponential polynomial}
		The holomorphic function $h$ is an exponential polynomial \cite[Chapter 3]{berenstein2012} in $-z^2$ with frequencies $\alpha_0 = 0$, $\alpha_1 = 2\pi^2\sigma_1^2$ and $\alpha_2 = 2\pi^2\sigma_2^2$ satisfying $\alpha_0<\alpha_1<\alpha_2$ due to assumptions on $\sigma_1$ and $\sigma_2$. It is \textit{normalized} since the coefficient of $0$-frequency equals $1$. A necessary condition for $h$ for being \textit{factorizable} \cite[Remark~3.1.5, p.~201]{berenstein2012} is that parameters $\sigma_1$ and $\sigma_2$ are taken so that it is \textit{simple}. By definition \cite[Definition 3.1.4, p. 201]{berenstein2012}, $h$ is simple if $\alpha_1$ and $\alpha_2$ are commensurable, i.e., $\alpha_1/\alpha_2$ is rational, which is equivalent to $\sigma_1^2/\sigma_2^2$ is rational. 
	\end{remark}
	\begin{remark}\label{rmk::variances of DOG for linear BT experiments}
		Thanks to Remark~\ref{rmk::exponential polynomial}, without loss of generality, we can set parameters in the kernel $\omega$ in \eqref{eq::connectivity} are such that $\kappa=1$, $\sigma_1=1/\pi\sqrt{2}$ and $\sigma_2 = \sigma_1\sqrt{2}$. We also let $\mu=1$.
	\end{remark}
	
	Using Theorem~\ref{thm::cplmt-BT}, and arguing similarly as in the proof of \cite[Proposition~5.12.]{tamekue2024mathematical}, we can prove that $a_1$ has a discrete and countable set of zeroes in $(0, +\infty)$, under the considerations in Remark~\ref{rmk::variances of DOG for linear BT experiments}.
	
	\begin{proof}[Proof of Theorem~\ref{thm::zero-level set linear BT}]
		To complete the proof of Theorem~\ref{thm::zero-level set linear BT}, it suffices to consider Lemma~\ref{lem::zero-level set linear BT 2}, Theorem~\ref{thm::cplmt-BT} and observe that $a_L$ given by \eqref{eq::a in lineat regime BT} satisfies \eqref{eq::zero-level set linear BT}.
	\end{proof}
	A consequence of Theorem~\ref{thm::zero-level set linear BT} is the following.
	\begin{corollary}
		Assume $\mu<\mu_{0}$ and that the response function $f$ is linear. Then,
		% that the response function $f$ in \eqref{eq::NF-intro} is linear. Under the assumption $\mu<\mu_0$,
		Equation \eqref{eq::NF-intro} does not reproduce Billock and Tsou's experiments associated with a sensory input consisting of a pattern of horizontal stripes localized in the left area in the cortex V1.
	\end{corollary}
	\begin{proof}
		Given that the sensory input in equation \eqref{eq::NF-intro} is a pattern consisting of horizontal stripes localized in the left area in the cortex V1, Theorem~\ref{thm::zero-level set linear BT} shows that the corresponding stationary state consists of a mixture of patterns of horizontal and vertical stripes in the right area in V1 instead of vertical stripes only, as Billock and Tsou reported.
	\end{proof}

	\subsection{Unreproducibility of Billock and Tsou's experiments with certain nonlinear response functions}\label{ss::unreproducibility of BT with certain nonlinearities}
	As we recalled in Section~\ref{ss::Assumption on parameters}, the numerical results provided in \cite[Fig.~8]{tamekue2023} suggest that a complex interplay of excitatory and inhibitory activity is required to model complex phenomena like Billock and Tsou's experiments using the Amari-type neural fields \eqref{eq::NF-intro}. In particular, they suggest adopting a nonlinear function $f_{m,\alpha}$ that allows for positive and negative values but is not odd, breaking the symmetry between excitatory and inhibitory influences. More precisely, \cite[Fig.~8]{tamekue2023} suggests that the stationary output of \eqref{eq::NF-intro} computed with the following response functions does not capture the essential features of visual illusions reported by Billock and Tsou. For $s\in\R$, they are given by:
	\begin{itemize}
		\myitem[\textbf{(NL1)}]\label{it:nl1} Nonnegative (wholly excitatory influence) nonlinearities:
		\begin{equation}\label{eq::nonnegative}
			f_{0,\alpha}(s) = \max(0,\min(1, \alpha s)),\qquad0<\alpha<\infty,
		\end{equation}
		\myitem[\textbf{(NL2)}]\label{it:nl2} Odd (balanced inhibitory and excitatory influence) nonlinearities:
		\begin{equation}\label{eq::odd}
			f_{1,\alpha}(s) = \max(-1,\min(1,\alpha s)),\qquad0<\alpha<\infty,
		\end{equation}
		\myitem[\textbf{(NL3)}]\label{it:nl3} Nonlinearities with a strong excitatory influence and a weak slope:
		\begin{equation}\label{eq::unbalanced and weak slope}
			f_{m,\alpha}(s) = \max(-m,\min(1,\alpha s)),
			% = \begin{cases}
				% 	1,\quad\mbox{if}\quad s\ge \frac 1\alpha,\cr
				% 	\alpha s,\quad\mbox{if}\quad-\frac m\alpha\le s\le \frac 1\alpha,\cr
				% 	-m,\quad\mbox{if}\quad s\le -\frac m\alpha,
				% \end{cases}
			\qquad 0<\alpha< m\le 1,
		\end{equation}
		\myitem[\textbf{(NL4)}]\label{it:nl4} Nonlinearities with a strong inhibitory influence and a weak slope:
		\begin{equation}\label{eq::semilinear and weak slope}
			f_{m,\alpha}(s) = \max(-m,\min(1,\alpha s)),
			% = \begin{cases}
				% 	1,\quad\mbox{if}\quad s\ge \frac 1\alpha,\cr
				% 	\alpha s,\quad\mbox{if}\quad-\frac m\alpha\le s\le \frac 1\alpha,\cr
				% 	-m,\quad\mbox{if}\quad s\le -\frac m\alpha,
				%                                        \end{cases}
			\qquad  0<\alpha< 1<m.
		\end{equation}
	\end{itemize}
	
	%\footnote{Since $f_{m,\alpha}(-\infty)=-m$, $f_{m,\alpha}(\infty)=1$ and $m<1$, the word \textit{unbalanced} is used to highlight that given a cortical activity $a(x,t)$ of a V1's spiking neuron at position $x\in\R^2$ at time $t\ge 0$, the excitatory output $f_{m,\alpha}(a(x,t))\ge 0$ is not compensate by inhibitory output $f_{m,\alpha}(a(x,t))\le 0$.}

	This section aims to provide analytical insight into why the Amari-type neural fields \eqref{eq::NF-intro} do not model Billock and Tsou's experiments when the response function is taken to be one of \ref{it:nl2}-\ref{it:nl4}.  The main focus will be on the study of the qualitative properties in terms of the zero-level set of the stationary solution to \eqref{eq::NF-intro} obtained with each of these nonlinearities when the sensory input is taken as $I_L$ defined in \eqref{eq::sensory inputs}. We are currently unable to treat the case \ref{it:nl1}.
	
	The first theorem of this section is the following.
	\begin{theorem}\label{thm::unreproducibility strong inhibitory or excitatory and weak slope}
		% Assume that the sensory input $I\in L^\infty(\R^2)$ in Equation~\eqref{eq::NF-intro}.
		If $\mu<\mu_0$ and $\min(1,m)\alpha^{-1}\ge\|I\|_\infty\left(1-\mu/\mu_0\right)^{-1}$, then the stationary solution $a_{m,\alpha}$ to \eqref{eq::NF-intro} with the response function $f_{m,\alpha}$ is the solution to \eqref{eq::NF-intro} with the linear response function with slope $\alpha>0$. In particular, if $I=I_L$, the nonlinear response functions \ref{it:nl3} and \ref{it:nl4} do not model Billock and Tsou's experiments.
	\end{theorem}
	\begin{proof}
		If $\mu<\mu_0$ then $a_{m,\alpha}\in L^\infty(\R^2)$ is the unique solution to $a_{m,\alpha} = I_L+\mu\omega\ast f_{m,\alpha}(a_{m,\alpha})$ thanks to Theorem~\ref{thm::existence of stationary input}. Recall from Theorem~\ref{thm::boundness of the solution} that $\|a_{m,\alpha}\|_\infty\le \|I\|_\infty(1-\mu/\mu_0)^{-1}$. If $\min(1,m)\alpha^{-1}\ge\|I\|_\infty\left(1-\mu/\mu_0\right)^{-1}$, then for a.e. $x\in\R^2$, one has 
		\[
		-\frac{m}{\alpha}\le a_{m,\alpha}(x)\le\frac{1}{\alpha}.
		\]
		Therefore, $f_{m,\alpha}(a_{m,\alpha}(y)) = a_{m,\alpha}(y)$ for a.e. $y\in\R^2$, and $a_{m,\alpha}\in L^\infty(\R^2)$ solves the stationary equation with a linear response function with slope $\alpha>0$. Finally, to prove the last part of the theorem, it suffices to observe that $\|I_L\|_\infty=1$ and $ (1-\mu/\mu_0)^{-1}>1$, which implies that $\min(1,m)>\alpha$, and then $\alpha<m\le 1$ or $\alpha<1<m$. The result then follows at once thanks to the first part of the theorem and Theorem~\ref{thm::zero-level set linear BT}.
	\end{proof}
	\begin{figure}[t]
		\centering
		\includegraphics[width=.65\linewidth]{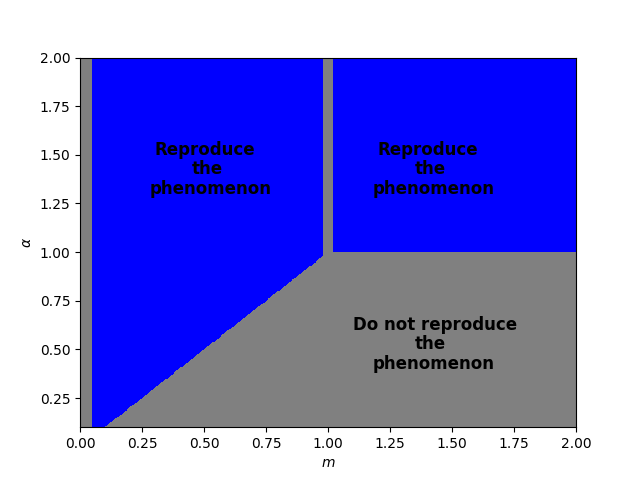}
		\vspace{-0.5cm}
		\caption{Summary on the ranges of parameters $m\ge 0$ and $\alpha>0$ where the nonlinearity $f_{m,\alpha}$ reproduces the phenomenon or not, for a funnel-like stimulus having the cortical representation defined in \eqref{eq::sensory inputs}. The parameters reproduce the phenomenon in \textit{blue}, and in \textit{grey}, they don't reproduce it. 
		}
		\label{fig:reproducibility}
	\end{figure}
	\begin{remark}
		Observe that Theorem~\ref{thm::unreproducibility strong inhibitory or excitatory and weak slope} also accounts for the case of $m=1$ and $\alpha<1$. This means that the odd nonlinearity $f_{1,\alpha}$ of \ref{it:nl2} with $0<\alpha<1$ does not model Billock and Tsou's experiments. It, therefore, remains to prove that the odd nonlinearity $f_{1,\alpha}$ with $\alpha\ge 1$  does not model Billock and Tsou's experiments.
	\end{remark}
	Fortunately, for all $\alpha\ge 0$, the odd nonlinearity $f_{1,\alpha}$ of \ref{it:nl2} satisfies all the hypotheses of Proposition~\ref{pro::general when f bounded and odd BT}, taken as the response function in \eqref{eq::NF-intro}. One, therefore, has the following result. See, for instance, Figure~\ref{fig:left horizontal stripes odd sigmoid} for a numerical visualization.
	% and the nonnegative nonlinearity $f_{0,\alpha}$ with $0<\alpha<\infty$ do not replicate Billock and Tsou's experiments.
	
	\begin{proposition}
		% Assume the neurones' response function is the odd nonlinearity $f_{1,\alpha}$.
		Under the assumption $\mu<\mu_0/2$, Equation \eqref{eq::NF-intro} with response function \ref{it:nl2} does not reproduce Billock and Tsou's experiments associated with a sensory input consisting of a pattern of horizontal stripes localized in the left area in the cortex V1.
	\end{proposition}
	\begin{proof}
		Given that the sensory input in Equation \eqref{eq::NF-intro} is a pattern consisting of horizontal stripes localized in the left area in the cortex V1, Proposition~\ref{pro::general when f bounded and odd BT} shows that the corresponding stationary state consists of a mixture of patterns of horizontal and vertical stripes in the right area in V1 instead of vertical stripes only, as Billock and Tsou reported.
	\end{proof}
	
	% Observe that we did not treat the nonnegative nonlinearity \ref{it:nl1}. We refer to Section~\ref{s::concluding remarks and discussion} for a discussion of this case.
	
	Summing up, the results in this section provide a complete theoretical investigation of Billock and Tsou's experiments modeling by Equation \eqref{eq::NF-intro} with response function $f_{m,\alpha}$, except for the range $m\neq 1$ and $\alpha\ge \min(m, 1)$. Although outside of the scope of this work, we observe that thanks to Corollary~\ref{coro::BT NC}, the study of this range can be reduced to considering the semilinear response function $f_{\infty,\alpha}$ as defined in \eqref{eq::f infinity alpha}. 
	
	{For fixed parameters $\sigma_1>0$, $\sigma_2>0$ and $\kappa>0$ in the kernel $\omega$, and $\mu>0$ chosen such that $\mu<\mu_0$, we summarize in Figure~\ref{fig:reproducibility} the ranges of parameters $m\ge 0$ and $\alpha>0$ for which equation \eqref{eq::NF-intro} with the response function $f_{m,\alpha}$ reproduce the phenomenon or not, for a funnel-like stimulus having the cortical representation defined in \eqref{eq::sensory inputs}.
	}
	
	\section{Numerical analysis and experiments}
	\label{s::numerics}
	% \label{s::numerical results for Billock and Tsou's experiments}
	
	In this section, we present a numerical scheme for the approximation of stationary solutions of \eqref{eq::NF-intro} and analyze its convergence. Finally, we present some numerical experiments obtained using this scheme.
	
	% Due to the difficulty of carrying out this study analytically because of the nonlinear nature of the equation, as a first attempt, we present a numerical analysis-type of argument. 
	% that should ``theoretically'' justify why  Equation \eqref{eq::NF-intro} with the aforementioned nonlinear response functions replicate Billock and Tsou experiments.
	
	\subsection{Analysis of the numerical scheme}\label{s::Analysis of NS}
	In this section, for the sake of generality, we assume that the response function $f$ satisfies the assumptions in Definition~\ref{def:response-fct}. 
	Given a sensory input $I\in L^\infty(\R^2)$, we compute numerical solutions $a_{n,h,M}:\Z^2\to \R$ depending on three parameters $h>0$, $M>0$, and $n\in \N$. These are obtained via the following iterative procedure, where $(i,j)\in \Z^2$:
	\begin{equation}\label{eq::sequence of functions numerical}
		\begin{split}
			a_{0,h,M}(i,j) &= I(ih,jh),\\
			a_{n+1,h,M}(i,j) &= I(ih,jh)+\mu h^2\sum_{p,q=-M}^{M}\omega(ph,qh)f(a_{n,h}(i-p,j-q)).
		\end{split}
	\end{equation}
	% The theoretical results presented below aim at estimating the error between $a_{n,h,M}$ and the exact stationary solution $a_I$.
	We start by presenting the following error estimate, whose proof is quite technical and is presented in Appendix~\ref{s::proof-of-numerical-error}.
	
	\begin{theorem}[Numerical error estimate]
		\label{thm:numerical-error}
		Assume that $\mu<\mu_0$ and let the sensory input $I$ be given by
		\begin{equation}
			I(x_1,x_2) = I_1(x_1,x_2)H(\theta_L - x_1)+I_2(x_1,x_2)H(x_1-\theta_L)
		\end{equation}
		where $\theta_L\in\mathbb{R}$, and $I_1,I_2\in L^\infty(\R^2)$ are globally Lipschitz continuous. Then, for any $\eta\in(\mu/\mu_0,1)$ there exists $h_0>0$ such that for every $h<h_0$ it holds
		\begin{equation}
			\sup_{(i,j)\in \Z^2}| a_I(ih,jh)-a_{n,h,M}(i,j)| = O(h)+O\left(\eta^n\right)+O\left(e^{-\frac{M^2h^2}{2\sigma_2^2}}\right)
		\end{equation}
		where the $O(\cdot)$'s only depend on $\mu$, $\eta$ $\alpha$, $\omega$, $\|I\|_\infty$, and the Lipschitz constants of $I_1$ and $I_2$.
	\end{theorem}
	
	\begin{remark}
		The only part of the proof where the Lipschitz continuity assumption on the sensory input in Theorem~\ref{thm:numerical-error} is needed is to control the error introduced by the discretization of the integral term of \eqref{eq::NF-intro}. 
		It is, however, easy to see that the argument of proof can be adapted to more general sensory inputs $I$, under appropriate assumptions on the region where $I$ is not Lipschitz continuous.
	\end{remark}
	
	\begin{remark}
		It is immediate from Theorem~\ref{thm:numerical-error} that to have numerical convergence to the exact stationary solution $a_I$, one has to choose $h\to 0$, $n\to +\infty$, and $M$ such that $Mh\to +\infty$.
	\end{remark}
	
	To compare the zero level-set of the exact solution with their numerical approximations, we introduce the following approximated zero level-set for $a_I$: 
	\begin{equation}
		\mathcal{Z}_{a_I}^\varepsilon = \{ x\in\R^2 \mid |a_I|< \varepsilon \}, \qquad \varepsilon>0.
	\end{equation}
	In order to define a numerical approximation of the above, 
	for a sensory input $I\in L^\infty(\mathbb{R}^2)$ as in Theorem~\ref{thm:numerical-error}, we let $\Omega_\pm=\{x\in\R^2\mid \pm(x_1-\theta_L)>0\}$. Then, for $(n,h,M)\in \N\times \R_+\times \R_+$, we define
	\begin{eqnarray}
		\mathcal{Z}_{n,h,M}^{\varepsilon,\pm} &=& \{ x\in \Omega_\pm \mid \exists (i,j)\in \Z^2 \text{ s.t.\ } \\ \nonumber 
		& & \qquad\qquad\qquad (ih,jh)\in\Omega_\pm,\,|x-(ih,jh)|< h/2 \text{ and } |a_{n,h,M}(i,j)|< \varepsilon \}\\
		\mathcal{Z}_{n,h,M}^{\varepsilon} &=& \mathcal{Z}_{n,h,M}^{\varepsilon,-}\cup \mathcal{Z}_{n,h,M}^{\varepsilon,+}.  
	\end{eqnarray}
	We have the following result, which guarantees the convergence of the numerical approximations of the zero level-set to the exact set $\mathcal Z_{a_I}$.
	
	\begin{theorem}
		Under the same assumptions as in Theorem \ref{thm:numerical-error}, for any $\varepsilon\in(0,1/2)$ it holds
		\begin{equation}
			\label{eq:inclusions}
			\mathcal{Z}_{n,h,M}^{\varepsilon/2} \subset \mathcal{Z}_{a_I}^{\varepsilon}
			\subset \mathcal{Z}_{n,h,M}^{2\varepsilon} 
		\end{equation}
		for any $(n,h,M)\in \N\times\R_+\times \R_+$ such that, for some constant $C>0$ depending only on $\mu$, $\alpha$, $\omega$, $\|I\|_\infty$ and the Lipschitz constants of $I_1$ and $I_2$, it holds
		\begin{equation}
			\label{eq:estimates-nhm}
			h \le C \varepsilon,
			\qquad
			n \ge -C \log\varepsilon,
			\qquad
			M \ge -C\frac{\log \varepsilon}{\varepsilon^2}.
		\end{equation}
	\end{theorem}
	
	\begin{proof}
		By Theorem~\ref{thm:numerical-error}, there exists $(n_0,h_0,M_0)\in \N\times \R_+\times \R_+$ such that for any $n>n_0$, $h<h_0$, and $M>M_0$, we have that 
		\begin{equation}
			\label{eq:zero1}
			\sup_{(i,j)\in \Z^2}| a_I(ih,jh)-a_{n,h,M}(i,j)| < \frac{\varepsilon}2.
		\end{equation}
		The estimate \eqref{eq:estimates-nhm} immediately follows choosing, e.g., $\eta=(1+\mu/\mu_0)/2$.
		Moreover, by Lipschitz continuity of $a_I$ on $\Omega_+\cup\Omega_-$, which is guaranteed by Proposition~\ref{pro::Lipschitz continuity}, up to reducing $h_0$ (i.e., reducing $C>0$), for all $(i,j)\in \Z^2$ with $(ih,jh)\in\Omega^\pm$ and $x\in \Omega^\pm$ such that $|x-(ih,jh)|\le h/2$, we have
		\begin{equation}
			\label{eq:zero2}
			| a_I(x)-a_{I}(ih,jh)| < \frac{\varepsilon}2.
		\end{equation}
		Combining \eqref{eq:zero1} and \eqref{eq:zero2}, one easily obtains \eqref{eq:inclusions}, completing the proof of the statement.
	\end{proof}
	
	\subsection{Simulations for Billock and Tsou experiments}\label{s::Simulations}
	
	The numerical implementation is obtained using the Julia toolbox from \cite{tamekue:tel-04230895}, which implements the numerical scheme presented above. 
	These experiments have been obtained with the parameters:
	\begin{equation}
		n = 10^2, \qquad h = 10^{-2}, \qquad M = 10^3.
	\end{equation}
	{ The sensory input is taken to be localized either on the left or on the right part of the cortical space. In the first case we let $I(x_1,x_2)=\cos(2\pi\lambda x_2)H(\theta-x_1)$, while in the second case $I(x_1,x_2)=\cos(2\pi\lambda x_2)H(x_1-\theta)$ for $\lambda>0$ and $\theta>0$. The choice of the input is precised in the captions, while we collect in Table~\ref{tab:parameters} the values of the parameters.}
	
	\begin{table}%[]
		\centering
		\begin{tabular}{|l||c|c||c|c||c|c|c||c|c|}
			\hline 
			& \multicolumn{2}{|c||}{$I$} & \multicolumn{2}{|c||}{$f$} & \multicolumn{3}{|c||}{$\omega$} & & \\
			& $\lambda$ & $\theta$ & $m$ & $\alpha$ & $\sigma_1$ & $\sigma_2$ & $\kappa$ & $\mu$ & $\mu_0$ \\
			\hline 
			\hline 
			Figure~\ref{fig:Billock-funnel_fovea}  & $0.4$ & $5$  &  $0.2$ & $0.5$ & $1/\pi\sqrt{2}$ & $1/\pi$ & $1.2$ & $1.5$ & $1.92$ \\
			Figure~\ref{fig:Billock-funnel_surround}  & $0.6$ & $2$  &  $0.2$ & $0.8$ & $1/\pi$ & $\sqrt{2}/\pi$ & $1$ & $1.2$ & $2$ \\
			Figure~\ref{fig:billock-funnel}  & $1.25$ & $3$ & $0.5$ & $1.5$ & $1/\pi\sqrt{2}$ & $1/\pi$ & $1.2$ & $1.5$ & $1.92$ \\
			Figure~\ref{fig:billock-funnel-surround} & $1$ & $2$ & $\infty$ & $5$ & $1/\pi$ & $\sqrt{2}/\pi$ & $1$ & $1.2$ & $2$ \\ 
			\hline
		\end{tabular}
		\caption{Parameters used in the presented numerical simulation.}
		\label{tab:parameters}
	\end{table}
	
	We exhibit in Figures~\ref{fig:Billock-funnel_fovea} and \ref{fig:Billock-funnel_surround} a numerical reproduction of Billock and Tsou's experiments, in the sense that the stripes' frequency is similar to that used in the experiment, for a funnel-like stimulus localized in the fovea and the peripheral visual field. In V1, we have a pattern of black/white horizontal stripes in the \textit{left} (respectively \textit{right}) area in V1 and white in the \textit{right} (respectively \textit{left}) area in V1. 
	We also exhibit in Figures~\ref{fig:billock-funnel} and \ref{fig:billock-funnel-surround} a numerical experiment where the stripes' frequency is not the one of Billock and Tsou's experiments. 
	
	Observe that each output pattern exhibited in Figures~\ref{fig:Billock-funnel_fovea}--\ref{fig:billock-funnel-surround} captures all the essential features of the after-image reported by Billock and Tsou at the level of V1. It suffices to apply the inverse retino-cortical map to each output pattern to obtain such images at the retina level. See, for instance, \cite{tamekue2023}.
	
	\begin{figure}[t]
		\centering
		% First image with TikZ square contour
		\begin{minipage}{0.45\textwidth}
			\centering
			\begin{tikzpicture}[remember picture]
				\node[anchor=south west,inner sep=0] (image1) at (0,0) {
					\includegraphics[width=0.8\textwidth]{BT-stimulus_fovea-V1.png}
				};
				\begin{scope}[x={(image1.south east)},y={(image1.north west)}]
					\draw[black, thick, opacity=0.1] (0,0) rectangle (1,1);
				\end{scope}
			\end{tikzpicture}
		\end{minipage}
		\hfill
		% Second image with TikZ square contour
		\begin{minipage}{0.45\textwidth}
			\centering
			\begin{tikzpicture}[remember picture]
				\node[anchor=south west,inner sep=0] (image2) at (0,0) {
					\includegraphics[width=0.8\textwidth]{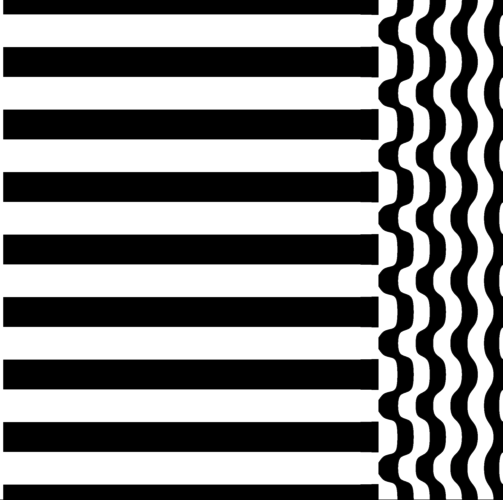}
				};
				\begin{scope}[x={(image2.south east)},y={(image2.north west)}]
					\draw[black, thick, opacity=0.1] (0,0) rectangle (1,1);
				\end{scope}
			\end{tikzpicture}
		\end{minipage}
		% Arrow with text (above) and math (below), drawn after both images have been placed
		\begin{tikzpicture}[remember picture, overlay]
			\draw[->, thick] ([xshift=5pt]image1.east) -- ([xshift=-5pt]image2.west)
			node[midway, above, font=\small] {via the stationary}
			node[midway, below, font=\small] {equation~\eqref{eq::SS}};
		\end{tikzpicture}
		\caption{{\textit{Left:}  sensory input $I(x_1, x_2) = \cos(2\pi\lambda x_2)H(\theta-x_1)$. \textit{Right:} corresponding stationary output.}}
		\label{fig:Billock-funnel_fovea}
	\end{figure}
	\begin{figure}%[ht]
		\centering
		% First image with TikZ square contour
		\begin{minipage}{0.45\textwidth}
			\centering
			\begin{tikzpicture}[remember picture]
				\node[anchor=south west,inner sep=0] (image1) at (0,0) {
					\includegraphics[width=0.8\textwidth]{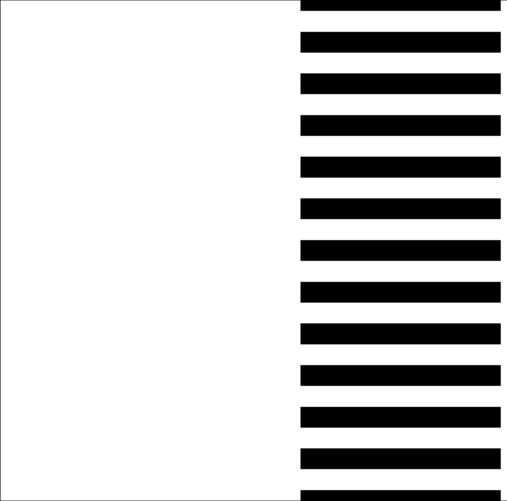}
				};
				\begin{scope}[x={(image1.south east)},y={(image1.north west)}]
					\draw[black, thick, opacity=0.1] (0,0) rectangle (1,1);
				\end{scope}
			\end{tikzpicture}
		\end{minipage}
		\hfill
		% Second image with TikZ square contour
		\begin{minipage}{0.45\textwidth}
			\centering
			\begin{tikzpicture}[remember picture]
				\node[anchor=south west,inner sep=0] (image2) at (0,0) {
					\includegraphics[width=0.8\textwidth]{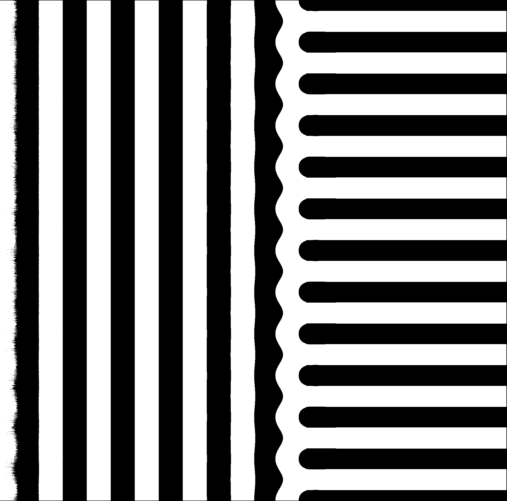}
				};
				\begin{scope}[x={(image2.south east)},y={(image2.north west)}]
					\draw[black, thick, opacity=0.1] (0,0) rectangle (1,1);
				\end{scope}
			\end{tikzpicture}
		\end{minipage}
		% Arrow with text (above) and math (below), drawn after both images have been placed
		\begin{tikzpicture}[remember picture, overlay]
			\draw[->, thick] ([xshift=5pt]image1.east) -- ([xshift=-5pt]image2.west)
			node[midway, above, font=\small] {via the stationary}
			node[midway, below, font=\small] {equation~\eqref{eq::SS}};
		\end{tikzpicture}
		\caption{{\textit{Left:}  sensory input $I(x_1, x_2) = \cos(2\pi\lambda x_2)H(x_1-\theta)$. \textit{Right:} corresponding stationary output.}}
		\label{fig:Billock-funnel_surround}
	\end{figure}
	\begin{figure}[t]
		\centering
		% First image with TikZ square contour
		\begin{minipage}{0.45\textwidth}
			\centering
			\begin{tikzpicture}[remember picture]
				\node[anchor=south west,inner sep=0] (image1) at (0,0) {
					\includegraphics[width=0.8\textwidth]{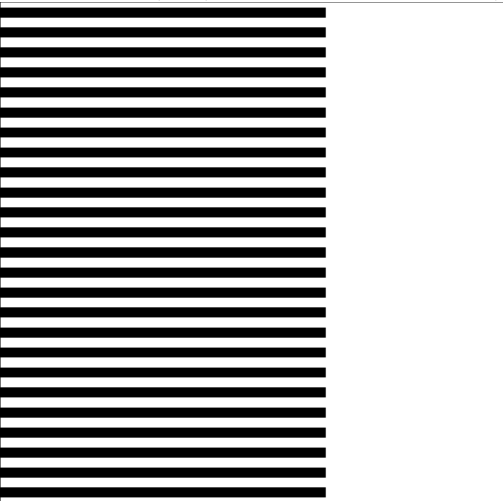}
				};
				\begin{scope}[x={(image1.south east)},y={(image1.north west)}]
					\draw[black, thick, opacity=0.1] (0,0) rectangle (1,1);
				\end{scope}
			\end{tikzpicture}
		\end{minipage}
		\hfill
		% Second image with TikZ square contour
		\begin{minipage}{0.45\textwidth}
			\centering
			\begin{tikzpicture}[remember picture]
				\node[anchor=south west,inner sep=0] (image2) at (0,0) {
					\includegraphics[width=0.8\textwidth]{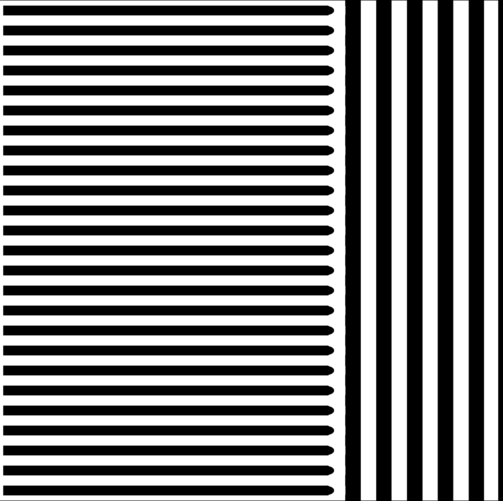}
				};
				\begin{scope}[x={(image2.south east)},y={(image2.north west)}]
					\draw[black, thick, opacity=0.1] (0,0) rectangle (1,1);
				\end{scope}
			\end{tikzpicture}
		\end{minipage}
		% Arrow with text (above) and math (below), drawn after both images have been placed
		\begin{tikzpicture}[remember picture, overlay]
			\draw[->, thick] ([xshift=5pt]image1.east) -- ([xshift=-5pt]image2.west)
			node[midway, above, font=\small] {via the stationary}
			node[midway, below, font=\small] {equation~\eqref{eq::SS}};
		\end{tikzpicture}
		\caption{{\textit{Left:}  sensory input $I(x_1, x_2) = \cos(2\pi\lambda x_2)H(\theta-x_1)$. \textit{Right:} corresponding stationary output.}}
		\label{fig:billock-funnel}
	\end{figure}
	
	\begin{figure}%[ht]
		\centering
		% First image with TikZ square contour
		\begin{minipage}{0.45\textwidth}
			\centering
			\begin{tikzpicture}[remember picture]
				\node[anchor=south west,inner sep=0] (image1) at (0,0) {
					\includegraphics[width=0.8\textwidth]{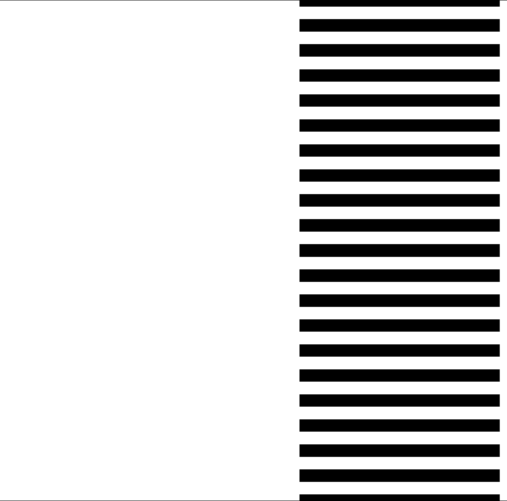}
				};
				\begin{scope}[x={(image1.south east)},y={(image1.north west)}]
					\draw[black, thick, opacity=0.1] (0,0) rectangle (1,1);
				\end{scope}
			\end{tikzpicture}
			% \caption{Funnel pattern in the centre of the visual field.}
			% \label{fig:fovea funnel odd sigmoid}
		\end{minipage}
		\hfill
		% Second image with TikZ square contour
		\begin{minipage}{0.45\textwidth}
			\centering
			\begin{tikzpicture}[remember picture]
				\node[anchor=south west,inner sep=0] (image2) at (0,0) {
					\includegraphics[width=0.8\textwidth]{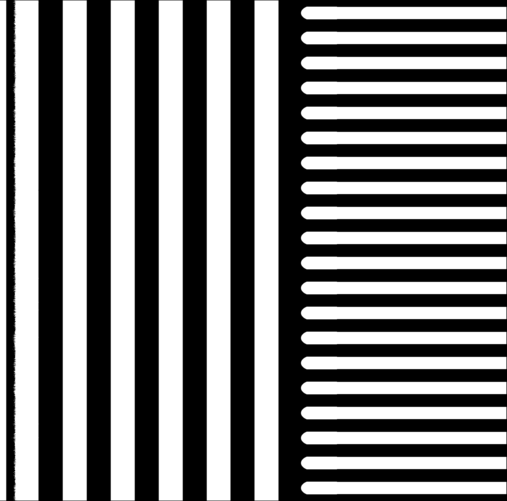}
				};
				\begin{scope}[x={(image2.south east)},y={(image2.north west)}]
					\draw[black, thick, opacity=0.1] (0,0) rectangle (1,1);
				\end{scope}
			\end{tikzpicture}
			% \caption{Horizontal stripes in the left area of V1.}
			% \label{fig:left horizontal stripes odd sigmoid}
		\end{minipage}
		% Arrow with text (above) and math (below), drawn after both images have been placed
		\begin{tikzpicture}[remember picture, overlay]
			\draw[->, thick] ([xshift=5pt]image1.east) -- ([xshift=-5pt]image2.west)
			node[midway, above, font=\small] {via the stationary}
			node[midway, below, font=\small] {equation~\eqref{eq::SS}};
		\end{tikzpicture}
		\caption{{\textit{Left:}  sensory input $I(x_1, x_2) = \cos(2\pi\lambda x_2)H(x_1-\theta)$. \textit{Right:} corresponding stationary output.}}
		\label{fig:billock-funnel-surround}
	\end{figure}
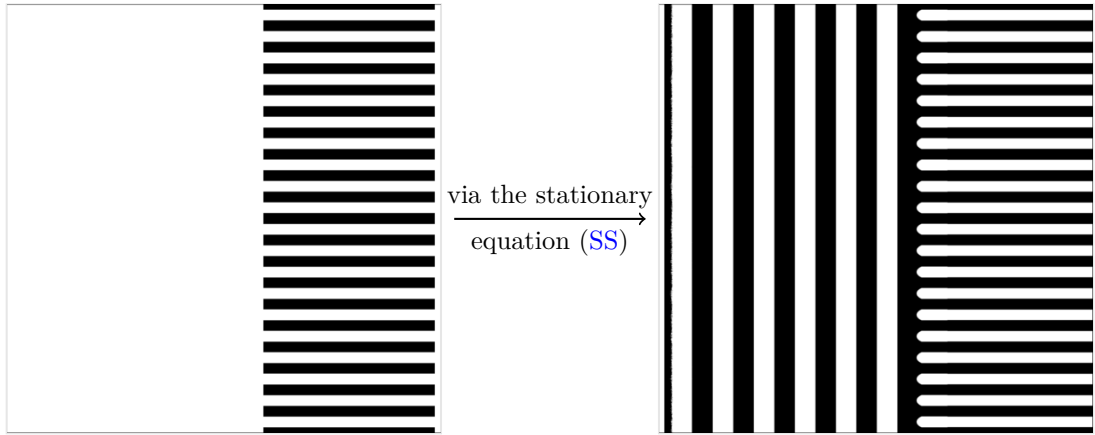

	\section{Concluding remarks and discussion}\label{s::concluding remarks and discussion}
	In this paper, we provided a mechanistic modeling of the psychophysical experiments reported by Billock and Tsou \cite{billock2007} using neural dynamics of Amari-type that models the average membrane potential of neuronal activity of spiking neurons
	in the primary visual cortex (V1). We focused on the case where intra-neural connectivity is smaller than the threshold where specific geometrical patterns spontaneously arise in the absence of sensory input from the retina. We considered, in particular, visual stimuli consisting of regular funnel patterns localized in the fovea or peripheral visual field. 
	
	Firstly, the retino-cortical map between the visual field and V1 allowed us to model these visual stimuli as patterns of horizontal stripes localized in the left or right area of V1 that we incorporated as sensory inputs in the neural fields equation. 
	Then, through complex and harmonic analysis tools, we have shown that when the neuronal response function of V1 is linear, the output pattern of the equation does not capture the V1 representation of the after-images reported by Billock and Tsou, suggesting that the phenomenon is wholly nonlinear. 
	Next, we dived into the study of nonlinear response functions for which the corresponding output patterns of the equation qualitatively capture, at the level of V1, the essential features of the after-images reported by Billock and Tsou. 
	Through this study, we have analytically shown that nonlinear response functions with either balanced inhibitory and excitatory influence or a strong excitatory influence and weak slope or a strong inhibitory influence and weak slope do not reproduce the phenomenon. This suggests that a complex interplay between excitatory and inhibitory influences \cite{haider2006neocortical,shu2003turning} is required for the neural fields equation to model the psychophysical observations reported by Billock and Tsou \cite{billock2007} for a funnel pattern visual stimulus localized either in the fovea or peripheral visual field. Finally, we presented numerical experiments showing that nonlinear response functions other than those enumerated previously can reproduce the phenomenon.
	
	While much remains to be understood about the mechanisms underlying Billock and Tsou's psychophysical observations, our study provides valuable insights into how the primary visual cortex processes sensory information arising from localized regular funnel patterns in the visual field. In particular, this study supports the experimental finding { (see, e.g., \cite[Experiment~3-p.~8492 and Discussion-p.~8493]{billock2007})} suggesting that there is an orthogonal response in the unexcited region of V1 as a response to simple geometrical patterns from the retina that do not fill all the visual field or are not regular in shape. 
	
	We stress that the structure of the visual stimuli related to funnel patterns used by Billock and Tsou at the V1 level was crucial to obtaining the results presented in this paper. The same modeling regarding the tunnel pattern localized in the fovea or the peripheral visual field (see \cite[Fig.~3b and 3d]{billock2007}) will not yield the after-images reported by Billock and Tsou.
	Indeed, due to the rotational invariance of these tunnel patterns, the stationary solutions induced by the corresponding sensory inputs will be invariant with respect to translations in the second variable of V1 (see, e.g., \cite[Proposition~A.1]{tamekue2024mathematical}). In particular, this excludes the possibility of a funnel-like after-image in the unexcited region.
	
	In this work, we have focused on time-independent visual stimuli, which turned out to be enough to model (static) nonlocal perceptual phenomena associated with the funnel patterns under consideration. Studying pattern formation from spatiotemporal visual stimuli would be interesting in future work. As an open question related directly to the current study, it will be interesting to analytically show that a nonnegative response function (as, e.g., the response function \ref{it:nl1} of Section~\ref{ss::unreproducibility of BT with certain nonlinearities}), which models wholly excitatory or inhibitory influence, does not reproduce the phenomenon, as suggested by the numerical simulations exhibited in \cite[Fig.~8]{tamekue2023}. Moreover, finding a systematic analytical method for explicitly studying the qualitative properties of the output pattern (e.g., the structure of the zero level-set) would be valuable. The starting point could be to investigate the case of the semilinear response function $f_{\infty,\alpha}$ since numerical analysis arguments and simulations suggest that this nonlinearity reproduces the phenomenon.

	\appendix
	
	\section{Complementary results}\label{s::complement resluts}
	This section contains miscellaneous results used in the previous sections. We begin with the following Gronwall's lemma, see for instance \cite[Proposition~2.1]{emmrich1999discrete} for a proof.
	\begin{lemma}\label{lem::gronwall lemma}
		Assume that $u\in C([0,T); \R)$, $T\in(0,\infty)$ satisfies the integral inequality
		\begin{equation}
			u(t)\le u(0)+\int_{0}^{t}g(s)u(s)ds+\int_{0}^{t}h(s)ds,\qquad\;\mbox{on}\;[0, T),
		\end{equation}
		for some $0\le g\in L^1(0,T)$ and $h\in L^1(0, T)$. Then $u$ satisfies the pointwise estimate
		\begin{equation}
			u(t)\le u(0)e^{G(t)}+\int_{0}^{t}h(s)e^{G(t)-G(s)}ds,\qquad\forall t\in(0,T),
		\end{equation}
		where $G(t)=\int_{0}^{t}g(s)ds$.
	\end{lemma}

	\subsection{Explicit computations of the kernel \texorpdfstring{$K$}{K} of Section~\ref{ss::theoretical Billock and Tsou experiments linear}}
	The following result is used to prove that \eqref{eq::NF-intro} with a linear response function does not model Billock and Tsou's observations for a funnel pattern localized either in the fovea or peripheral visual field.  
	%The first is the following.
	\begin{theorem}\label{thm::cplmt-BT}
		Under the considerations of Remark~\ref{rmk::variances of DOG for linear BT experiments},
		% \begin{equation}\label{eq::solution of SS DoG}
			% 	b(x) = I(x)+(K\ast I)(x),\qquad\qquad\forall x\in\R^{*},
			% \end{equation}
		the kernel $K$ defined in \eqref{eq::noyau K} satisfies, for any $x\in\R^{*}$,
		\begin{eqnarray}\label{eq:: kernel K}
			\frac{\sqrt 3}{2\pi} K(x) &=&  \frac{e^{-2\pi m_0|x|}}{\sqrt{n_0^2+m_0^2}}\cos\left( 2\pi n_0|x|+\frac{4\pi}{3}  -\phi_0\right)\nonumber\\%\frac{2\pi}{\sqrt{3}}
			&& \qquad\qquad+\sum\limits_{k=1}^{\infty}\frac{e^{-2\pi m_k|x|}}{\sqrt{n_k^2+m_k^2}}\cos\left( 2\pi n_k|x|+\frac{4\pi}{3}  -\phi_k\right) \nonumber\\%\frac{2\pi}{\sqrt{3}}
			&&\qquad\qquad+ \sum\limits_{k=1}^{\infty}\frac{e^{-2\pi f_k|x|}}{\sqrt{f_k^2+e_k^2}}\cos\left( 2\pi e_k|x|+\frac{4\pi}{3}  -\theta_k\right).
			%\frac{e^{-2\pi m_0|x|}}{\sqrt{n_0^2+m_0^2}}\cos\left( 2\pi n_0|x|+\frac{4\pi}{3}  -\phi_0\right) + O\left( \frac{1}{|x|} \right),
		\end{eqnarray}
		Here, for any $k\in \N$, we have that $\phi_k,\theta_k\in \R$, and, letting $c_k = \sqrt{1+6k}$ and $d_k = \sqrt{-1+6k}$, we have
		\begin{equation}\label{eq:: m_k and n_k}
			m_k^2=\frac{1+\sqrt{1+\frac{\pi^2}{9}c_k^4}}{2},\qquad n_k^2 = \frac{-1+\sqrt{1+\frac{\pi^2}{9}c_k^4}}{2},\qquad k\in\N
		\end{equation}
		\begin{equation}\label{eq:: e_k and f_k}
			e_k^2=\frac{1+\sqrt{1+\frac{\pi^2}{9}d_k^4}}{2},\qquad f_k^2 = \frac{-1+\sqrt{1+\frac{\pi^2}{9}d_k^4}}{2},\qquad k\in\N.
		\end{equation}
	\end{theorem}
	\begin{proof}
		We recall that for $x_1\in\R$, one has 
		$$
		K(x_1) = \int_{-\infty}^{+\infty}e^{2i\pi x_1\xi}\frac{\widehat{\psi_1}(\xi)}{1-\widehat{\psi_1}(\xi)}d\xi,\qquad \widehat{\psi_1}(\xi) = e^{-(1+\xi^2)}- e^{-2(1+\xi^2)},\qquad \xi\in\R.
		$$
		We are looking for poles of the following meromorphic function
		\begin{equation}
			h(z)	=\frac{\widehat{\psi_1}(z)}{1-\widehat{\psi_1}(z)}e^{2i\pi x_1\xi},\qquad\qquad\widehat{\psi_1}(z) = e^{-(1+z^2)}-e^{-2(1+z^2)},\qquad z\in\C.
		\end{equation}
		By careful computations, one finds that the poles of $h$ in $\C$ are given by $F_{k,\ell}$, $\overline{F_{k,\ell}}$, $G_{k,\ell}$ and $\overline{G_{k,\ell}}$, where for $\ell\in\{0,1\}$, one has
		$$
		F_{k,\ell} = (-1)^\ell n_k+im_k,\qquad k\in \N,\qquad\mbox{and}\qquad G_{k,\ell} = (-1)^\ell f_k+ie_k,\qquad k\in \N^{*}
		$$
		where $m_k$ and $n_k$ are given by \eqref{eq:: m_k and n_k}, and  $e_k$ and $f_k$ are given by \eqref{eq:: e_k and f_k}. Then the residue of $h$ are given  for $\ell\in\{0,1\}$ by
		%	$$
		%	\Res(h,z_0) =-\frac{e^{2i\pi x_1z_0}}{2z_0(1-e^{-2(1+z_0^2)})}.
		%	$$
		\begin{equation}\label{eq::residue F(k,ell)}
			\Res(h,F_{k,\ell})= \frac{(-1)^\ell i\overline{F_{k,\ell}}e^{(-1)^\ell i\frac\pi 3}}{2\sqrt{3+\frac{\pi^2}{3}c_k^4}}e^{2i\pi x_1 F_{k,\ell}},\qquad  \Res(h,F_{k,\ell})=\overline{\Res(h,\overline{F_{k,\ell}})},\quad k\in\N
		\end{equation}
		\begin{equation}\label{eq::residue G(k,ell)}
			\Res(h,G_{k,\ell})= \frac{-(-1)^\ell i\overline{G_{k,\ell}}e^{-(-1)^\ell i\frac\pi 3}}{2\sqrt{3+\frac{\pi^2}{3}d_k^4}}e^{2i\pi x_1 G_{k,\ell}},\qquad  \Res(h,G_{k,\ell})=\overline{\Res(h,\overline{G_{k,\ell}})},\quad k\in\N^{*}.
		\end{equation}
		We now fix $x_1>0$, and we let  
		$$
		R_n:=\frac{\sqrt{\sqrt{1+\frac{\pi^2}{9}c_n^4}}+\sqrt{\sqrt{1+\frac{\pi^2}{9}d_n^4}}}{2},\qquad\qquad n\in\N^{*}.
		$$
		We consider the path $\Gamma_n$ straight along the real line axis from $-R_n$ to $R_n$ and then counterclockwise along a semicircle centred at $z=0$ in the upper half of the complex plane, $\Gamma_n = [-R_n,R_n]\cup C_n^{+}$, where $C_n^{+} = \{R_ne^{i\theta}\mid\theta\in[0,\pi]\}$. Then, by the residue Theorem, one has for all $n\in\N^{*}$,
		\begin{eqnarray}\label{eq::residue thm DoG Billock and Tsou}
			\int_{-R_n}^{R_n}h(\xi)d\xi+\int_{C_n^{+}}h(z)dz&=&2\pi i\sum\limits_{\ell=0}^{\ell=1}\sum\limits_{k=0}^{n-1}	\Res(h,F_{k,\ell})+2\pi i\sum\limits_{\ell=0}^{\ell=1}\sum\limits_{k=1}^{n-1}	\Res(h,G_{k,\ell})\nonumber\\
			&=& \frac{2\pi}{\sqrt{3}}\frac{e^{-2\pi m_0|x|}}{\sqrt{n_0^2+m_0^2}}\cos\left( 2\pi n_0|x|+\frac{4\pi}{3}  -\phi_0\right)+\nonumber\\
			&& \frac{2\pi}{\sqrt{3}}\sum\limits_{k=1}^{n-1}\frac{e^{-2\pi m_k|x|}}{\sqrt{n_k^2+m_k^2}}\cos\left( 2\pi n_k|x|+\frac{4\pi}{3}  -\phi_k\right)+\nonumber\\
			&&\frac{2\pi}{\sqrt{3}}\sum\limits_{k=1}^{n-1}\frac{e^{-2\pi f_k|x|}}{\sqrt{f_k^2+e_k^2}}\cos\left( 2\pi e_k|x|+\frac{4\pi}{3}  -\theta_k\right)
		\end{eqnarray}
		where $\phi_k:=\phi_k(m_k, n_k)\in\R$ and $\theta_k:=\theta_k(e_k, f_k)\in\R$ are such that
		$$
		\cos\phi_k = \frac{n_k}{\sqrt{m_k^2+n_k^2}},\qquad \sin\phi_k = \frac{m_k}{\sqrt{m_k^2+n_k^2}},\qquad k\in\N
		$$
		$$
		\cos\theta_k = -\frac{f_k}{\sqrt{e_k^2+f_k^2}},\qquad \sin\theta_k = \frac{e_k}{\sqrt{e_k^2+f_k^2}},\qquad k\in\N^{*}.
		$$
		Arguing similarly as in the proof of \cite[Theorem B.1. ]{tamekue2024mathematical} we can prove that 
		$$
		\int_{C_n^{+}}h(z)dz\xrightarrow[n\to+\infty]{}0.
		$$
		Finally, passing in the limit as $n\rightarrow+\infty$ in \eqref{eq::residue thm DoG Billock and Tsou} completes the proof.
	\end{proof}
	
	\subsection{Proof of Theorem~\ref{thm:numerical-error}}
	\label{s::proof-of-numerical-error}
	
	% \begin{proof}[Proof of Theorem~\ref{thm:numerical-error}]
		We start by noticing that 
		\begin{equation}
			\lim_{h\to 0}h^2\sum_{p,q=-\infty}^{\infty}|\omega(ph,qh)|  
			=
			\|\omega\|_1.
		\end{equation}
		Hence, one can take $h>0$ small enough, such that 
		\begin{equation}
			\label{eq:approx-mu-h-M}
			\mu \alpha h^2\sum_{p,q=-\infty}^{\infty}|\omega(ph,qh)| \le  \eta < 1.
		\end{equation}
		Consider the fixed point equation
		\begin{equation}
			\label{eq:fixed-point-eq}
			b(i,j) = I(ih,jh)+\mu h^2\sum_{p,q=-M}^{M}\omega(ph,qh)f(b(i-p,j-q)), \qquad b\in \ell^\infty(\Z^2).
		\end{equation}
		Thanks to \eqref{eq:approx-mu-h-M}, the contraction mapping principle ensures the existence and uniqueness of the solution $a_{h,M}$ to the above. In particular, it holds
		\begin{equation}
			\label{eq:ahm-anhm}
			\sup_{(i,j)\in \Z^2}| a_{h,M}(i,j)-a_{n,h,M}(i,j)| \le \frac{\eta^{n+1}}{1-\eta} \|I\|_\infty .
		\end{equation}
		
		Consider now the fixed point equation of the type \eqref{eq:fixed-point-eq} with $M=+\infty$. Thanks to \eqref{eq:approx-mu-h-M}, this equation admits a unique solution $a_h\in \ell^\infty(\Z^2)$ such that
		\begin{equation}
			\label{eq:ah}
			a_h(i,j) = I(ih,jh)+\mu h^2\sum_{p,q=-\infty}^{\infty}\omega(ph,qh)f(a_h(i-p,j-q)).
		\end{equation}
		We now claim that there exists a constant $C_\omega>0$ depending only on the parameters of the coupling kernel $\omega$ such that
		\begin{equation}
			\label{eq:ah-ahm}
			\sup_{(i,j)\in \Z^2}| a_{h}(i,j)-a_{h,M}(i,j)| \le \frac{\mu\alpha}{(1-\eta)^2} \|I\|_\infty C_\omega e^{-\frac{M^2 h^2}{2\sigma_2^2}}.
		\end{equation}
		First of all, observe that by \eqref{eq:ah} we have
		\begin{equation}
			\label{eq:norm-ah-boh}
			\sup_{(i,j)\in \Z^2}|a_h(i,j)|\le \frac{\|I\|_\infty}{1-\eta}.
		\end{equation}
		Next, for any $(i,j)\in\Z^2$, by \eqref{eq:fixed-point-eq} and \eqref{eq:ah}, we have
		\begin{equation}
			\label{eq:def-J12}
			a_{h}(i,j)-a_{h,M}(i,j) = J_1 + J_2
		\end{equation}
		where
		\begin{eqnarray}
			J_1 &=& \mu h^2\sum_{\max\{|p|,|q|\}\ge M+1} \omega(ph,qh)f(a_{h}(i-p,j-q)) \\
			J_2 &=& \mu h^2\sum_{p,q = -M}^M \omega(ph,qh)\bigg(f(a_{h}(i-p,j-q))-f(a_{h,M}(i-p,j-q))\bigg).
		\end{eqnarray}
		Using \eqref{eq:norm-ah-boh} and the fact that $f$ is globally $\alpha$-Lipschitz continuous, one has
		\begin{eqnarray}
			|J_1| &\le& \frac{\mu h^2\alpha\|I\|_\infty}{1-\eta} \sum_{\max\{|p|,|q|\}\ge M+1} |\omega(ph,qh)|\\
			\label{eq:J2}
			|J_2| &\le& \mu h^2\alpha \sup_{(i,j)\in\Z^2} |a_h(i,j)-a_{h,M}(i,j)| \sum_{p,q=-M}^M |\omega(ph,qh)|
		\end{eqnarray}
		It is then immediately seen that 
		\begin{equation}
			\label{eq:estimate-J1}
			|J_1| \le \frac{\mu h^2\alpha\|I\|_\infty}{1-\eta}  C_\omega \int_{M+1}^{\infty} e^{-\frac{x_1^2h^2}{2\sigma_2^2}}\,dx_1\int_{-\infty}^{\infty} e^{-\frac{x_2^2h^2}{2\sigma_2^2}}\,dx_2
			\le \frac{\mu\alpha\|I\|_\infty}{1-\eta}  {C_\omega}  e^{-\frac{M^2h^2}{2\sigma_2^2}}.
		\end{equation}
		Here, $C_\omega>0$ denotes possibly different constants only depending on $\omega$. As for $J_2$, we deduce from  \eqref{eq:approx-mu-h-M} and \eqref{eq:J2} that
		\begin{equation}
			|J_2|\le \eta \sup_{(i,j)\in \Z^2}|a_h(i,j)|.
		\end{equation}
		Collecting \eqref{eq:def-J12}, \eqref{eq:estimate-J1}, and the above completes the proof of the claim.
		
		We are now left to upper-bound $|a_I(ih,jh)-a_{h}(i,j)|$ for all $(i,j)\in \Z^2$ and $h$ small enough. To proceed, we define the squares $Q_{p,q} = (ph,(p+1)h)\times(qh,(q+1)h)\subset\R^2$ for $(p,q)\in\Z^2$. By definition of $a_I$ and $a_h$, one gets that for every $(i,j)\in\Z^2$
		\begin{equation}
			a_I(ih,jh)-a_h(i,j) = \mu \sum_{p,q=-\infty}^\infty  \left(K_{p,q}^1 + K_{p,q}^2 + K_{p,q}^3\right)
		\end{equation}
		where
		\begin{eqnarray}
			K_{p,q}^1&=&
			\int_{Q_{p,q}}\omega(y)\bigg(  f(a_I(ih-y_1,jh-y_2)) -  f(a_I((i-p)h,(j-q)h)) \bigg) \,dy\\
			K_{p,q}^2&=&
			\int_{Q_{p,q}}\bigg(\omega(y)-\omega(ph,qh)\bigg)  f(a_I((i-p)h,(j-q)h))  \,dy\\
			K_{p,q}^3&=&
			\int_{Q_{p,q}}\omega(ph,qh)\bigg(  f(a_I((i-p)h,(j-q)h))-f(a_h(i-p,j-q)) \bigg) \,dy.
		\end{eqnarray}
		By Theorem~\ref{thm::boundness of the solution} and the $\alpha$-Lipschitz continuity of $f$, it is immediate to see that
		\begin{equation}
			\label{eq:81}
			|K_{p,q}^2|\le \frac{\alpha\|I\|_\infty}{1-\frac{\mu}{\mu_0}} h^3 \max_{Q_{p,q}}|\nabla \omega|
			\quad\text{and}\quad
			|K^3_{p,q}|\le \alpha h^2 |\omega(ph,qh)| \sup_{(i,j)\in\Z^2}|a_I(ih,jh)-a_h(i,j)|.
		\end{equation}
		Observe that there exists $C_\omega\ge 1$ such that $\max_{Q_{p,q}}|\nabla \omega| \le C_\omega|\nabla\omega(ph,qh)|$, for every $(p,q)\in \Z^2$. Hence, it follows that there exists $C_\omega>0$ such that
		\begin{equation}
			\label{eq:85}
			\mu\sum_{p,q=-\infty}^\infty |K_{p,q}^2| \le C_\omega \frac{\mu\alpha\|I\|_\infty}{1-\frac{\mu}{\mu_0}} h.
		\end{equation}
		On the other hand, by \eqref{eq:approx-mu-h-M}, we have
		\begin{equation}
			\label{eq:86}
			\mu\sum_{p,q=-\infty}^\infty |K_{p,q}^3| \le \eta \sup_{(i,j)\in\Z^2}|a_I(ih,jh)-a_h(i,j)|.
		\end{equation}

		To estimate $K^1_{p,q}$, we start by noticing that, by construction, there exists $p_0\in\N$ such that $K^1_{p,q}\cap\{x_1=\theta_L\}\neq \emptyset$ if and only if $p=p_0$. In particular, $a_I$ is Lipschitz continuous on $Q_{p,q}$ if $p\neq p_0$ by Proposition~\ref{pro::Lipschitz continuity}, with Lipschitz constant upper-bounded by $D_I$ defined in \eqref{eq:DI} where the corresponding $L_I$ is equal to the maximum of the Lipschitz constants of $I_1$ and $I_2$. Hence, for every $(p,q)\in\Z^2$ $p\neq p_0$ we have
		\begin{equation}
			|K^1_{p,q}|\le \alpha D_I h \int_{Q_{p,q}}|\omega(y)|\,dy.
		\end{equation}
		It follows that
		\begin{equation}
			\label{eq:88}
			\mu\sum_{\substack{(p,q)\in\Z^2\\ p\neq p_0}} |K_{p,q}^1| \le \mu \alpha D_I h \|\omega\|_1.
		\end{equation}
		On the other hand, for every $(p_0,q)$, $q\in\Z$, we have
		\begin{equation}
			|K^1_{p_0,q}|\le \alpha \frac{\|I\|_\infty}{1-\frac{\mu}{\mu_0}}  \int_{Q_{p_0,q}}|\omega(y)|\,dy.
		\end{equation}
		Hence, there exists a constant $C_\omega>0$
		\begin{equation}
			\label{eq:90}
			\mu\sum_{q=-\infty}^\infty |K_{p_0,q}^1| \le  \frac{\mu\alpha\|I\|_\infty}{1-\frac{\mu}{\mu_0}} \int_{\{p_0 h\le y_1\le (p_0+1) h\}}|\omega(y)|\,dy 
			\le  \frac{\mu\alpha\|I\|_\infty}{1-\frac{\mu}{\mu_0}} C_\omega h.
		\end{equation}
		By collecting the estimates \eqref{eq:85}, \eqref{eq:86}, \eqref{eq:88}, and \eqref{eq:90}, we obtain that
		\begin{equation}
			\label{eq:aI-ah}
			\sup_{(i,j)\in\Z^2}|a_I(ih,jh)-a_h(i,j)| \le 
			\frac{C_\omega \mu \alpha}{(1-\eta)^2} \left( D_I + \|I\|_\infty \right)h.
		\end{equation}
		Finally, collecting \eqref{eq:ahm-anhm}, \eqref{eq:ah-ahm} and \eqref{eq:aI-ah} yields the statement.
		\qed
		% \end{proof}
	
	% Same-type analysis can be made to prove that the kernel $K$ has a countable and discrete set of zeroes in $(0, +\infty)$, as in the case of the kernel $K$ studied in Section~\textcolor{blue}{cite SIAM paper}, refer in particular to Proposition~\textcolor{blue}{cite SIAM paper}.

	% \begin{figure}%[b]{0.3\textwidth}
		%          \centering
		%         \animategraphics[controls={play, stop},width=\textwidth]{50}{folder-animate/output-}{1}{399}
		%          \caption{}
		% \end{figure}

	\printbibliography
	
\end{document}